\documentclass[journal]{IEEEtran}

\IEEEoverridecommandlockouts

\usepackage[noadjust]{cite}						
\usepackage[nolist,nohyperlinks,printonlyused]{acronym}

\usepackage{booktabs}							
\usepackage{tabularx}							

\usepackage{supertabular}
\usepackage{multirow}

\usepackage{nicefrac}							
\usepackage{siunitx}							
\usepackage{amsmath,amssymb,amsfonts}
\usepackage{gensymb}
\usepackage{mathtools}							

\usepackage{graphicx}
\usepackage{color}								
\usepackage{xcolor}
\usepackage{caption}
\usepackage{subcaption}
\usepackage{tikz}
\usepackage{pgfplots}
\usepackage{grffile}
\usepackage[european, americanvoltages, nooldvoltagedirection]{circuitikz}		
\usepackage{stfloats}

\usepackage[german,USenglish,UKenglish]{babel}
\usepackage{textcomp}							
\usepackage[T1]{fontenc}
\usepackage{bm}									
\usepackage{dsfont}								
\usepackage{lmodern}							

\usepackage{xspace}
\usepackage{enumerate}							
\usepackage{cancel}								
\usepackage[normalem]{ulem}						
\usepackage{amsthm}								
\usepackage{thmtools}							
\usepackage{hyphenat}							
\usepackage{todonotes}							

\usepackage{xifthen}							

\usetikzlibrary{external}
\usetikzlibrary{decorations.markings}
\usetikzlibrary{shapes,arrows}
\usetikzlibrary{calc}
\usetikzlibrary{angles}
\usetikzlibrary{positioning}
\usetikzlibrary{fit}
\usetikzlibrary{backgrounds}
\usetikzlibrary{mindmap,trees}

\pgfplotsset{compat=newest}
\usetikzlibrary{plotmarks}
\usetikzlibrary{arrows.meta}
\usetikzlibrary{bending}
\usetikzlibrary{shadows}
\usepgfplotslibrary{patchplots}

\usepackage{arydshln}					
\usepackage[colorlinks=true,linkcolor=black,anchorcolor=black,citecolor=black,filecolor=black,menucolor=black,runcolor=black,urlcolor=blue]{hyperref}

\usepackage{balance}

	\definecolor{purpleDark}{RGB}{118, 4, 205}
	\definecolor{purpleLight}{RGB}{186, 102, 250}
	
	\definecolor{blueDark}{RGB}{52, 78, 243}
	\definecolor{blueLight}{RGB}{118, 135, 244}
	\definecolor{redDark}{RGB}{197, 34, 0}
	\definecolor{redLight}{RGB}{255, 91, 57}
	\definecolor{yellowDark}{RGB}{255, 183, 0}
	\definecolor{yellowLight}{RGB}{255, 204, 77}
	\definecolor{greenDark}{RGB}{0, 143, 53}
	\definecolor{greenLight}{RGB}{42, 189, 97}
	
	\colorlet{greenFaint}{green!10!white}
	\colorlet{redFaint}{red!10!white}

	\definecolor{redText}{RGB}{222, 2, 10}
	\definecolor{orangeText}{RGB}{245, 86, 0}
	\definecolor{greenText}{RGB}{20,125,50}
	\definecolor{blueText}{RGB}{0, 114, 190}
	\definecolor{purpleText}{RGB}{115, 38, 146}
	\definecolor{pinkText}{RGB}{255, 107, 250}
	
	\definecolor{changesColour}{RGB}{52, 78, 243}

	\definecolor{ballblue}{rgb}{0.13, 0.67, 0.8}
	\definecolor{buff}{rgb}{0.94, 0.86, 0.51}
	\definecolor{bronze}{rgb}{0.93,0.53,0.18}
	
	\definecolor{matlabCol1}{rgb}{0.0000,0.4470,0.7410}
	\definecolor{matlabCol2}{rgb}{0.8500,0.3250,0.0980}
	\definecolor{matlabCol3}{rgb}{0.9290,0.6940,0.1250}
	\definecolor{matlabCol4}{rgb}{0.4940,0.1840,0.5560}
	\definecolor{matlabCol5}{rgb}{0.4660,0.6740,0.1880}
	\definecolor{matlabCol6}{rgb}{0.3010,0.7450,0.9330}
	\definecolor{matlabCol7}{rgb}{0.6350,0.0780,0.1840}
	\definecolor{matlabCol8}{rgb}{0.0000,0.0000,0.0000}
	
	\newcommand{\convexpath}[2]{
		[   
		create hullcoords/.code={
			\global\edef\namelist{#1}
			\foreach [count=\counter] \nodename in \namelist {
				\global\edef\numberofnodes{\counter}
				\coordinate (hullcoord\counter) at (\nodename);
			}
			\coordinate (hullcoord0) at (hullcoord\numberofnodes);
			\pgfmathtruncatemacro\lastnumber{\numberofnodes+1}
			\coordinate (hullcoord\lastnumber) at (hullcoord1);
		},
		create hullcoords
		]
		($(hullcoord1)!#2!-90:(hullcoord0)$)
		\foreach [
		evaluate=\currentnode as \previousnode using \currentnode-1,
		evaluate=\currentnode as \nextnode using \currentnode+1
		] \currentnode in {1,...,\numberofnodes} {
			let \p1 = ($(hullcoord\currentnode) - (hullcoord\previousnode)$),
			\n1 = {atan2(\y1,\x1) + 90},
			\p2 = ($(hullcoord\nextnode) - (hullcoord\currentnode)$),
			\n2 = {atan2(\y2,\x2) + 90},
			\n{delta} = {Mod(\n2-\n1,360) - 360}
			in 
			{arc [start angle=\n1, delta angle=\n{delta}, radius=#2]}
			-- ($(hullcoord\nextnode)!#2!-90:(hullcoord\currentnode)$) 
		}
	}

	\newcommand{\autorefapp}[1]{\hyperref[#1]{Appendix~\ref*{#1}}}
	
	\addto\extrasUKenglish{%
	}

	\addto\extrasUSenglish{%
	}

	\makeatletter
	\newcommand\autorefMulti[1]{\@first@ref#1,@}
	\def\@throw@dot#1.#2@{#1}
	\def\@set@refname#1{
		\edef\@tmp{\getrefbykeydefault{#1}{anchor}{}}%
		\xdef\@tmp{\expandafter\@throw@dot\@tmp.@}%
		\ltx@IfUndefined{\@tmp autorefnameplural}%
			{\def\@refname{\@nameuse{\@tmp autorefname}s}}%
			{\def\@refname{\@nameuse{\@tmp autorefnameplural}}}%
	}
	\def\@first@ref#1,#2{%
		\ifx#2@\autoref{#1}\let\@nextref\@gobble
		\else%
			\@set@refname{#1}
			\@refname~\ref{#1}
			\let\@nextref\@next@ref
		\fi%
		\@nextref#2%
	}
	\def\@next@ref#1,#2{%
		\ifx#2@ and~\ref{#1}\let\@nextref\@gobble
		\else, \ref{#1}
		\fi%
		\@nextref#2%
	}
	\makeatother

	\DeclareSIUnit{\pu}{pu}							
	\DeclareSIUnit{\VAR}{\volt\ampere{}R}			
	
	\newcommand{\addWithPreComma}[1]{%
		\if\relax #1\relax
		\else%
		,#1
		\fi%
	}
	\newcommand{\addWithPostComma}[1]{%
		\if\relax #1\relax
		\else%
		#1,
		\fi%
	}
	\newcommand{\addInParentheses}[1]{%
	\if\relax #1\relax
	\else%
		(#1)
	\fi%
	}

\newacro{PI}[PI]{proportional-integral}
\newacro{DDA}[DDA]{dynamic distributed averaging}
\newacro{DGU}[DGU]{distributed generation unit}

\newacro{EIP}[EIP]{equilibrium-independent passivity}
\newacro{OSEIP}[OS-EIP]{output strictly equilibrium-independent passivity}
\newacro{IFP}[IFP]{input-feedforward passive}
\newacro{OFP}[OFP]{output-feedback passive}
\newacro{IFOFP}[IF-OFP]{input-feedforward output-feedback passive}

\newacro{LMI}[LMI]{linear matrix inequality}

\newacro{ISS}[ISS]{input-to-state stability}

\newacro{ZSO}[ZSO]{zero-state observable}
\newacro{ZSD}[ZSD]{zero-state detectable}

\newtheorem{theorem}{Theorem}
\newtheorem{definition}[theorem]{Definition}
\newtheorem{proposition}[theorem]{Proposition}

\newtheorem{corollary}[theorem]{Corollary}
\newtheorem{remark}{Remark}
\newtheorem{objective}{Objective}
\newtheorem{assumption}{Assumption}

	\renewcommand{\vec}[1]{\bm{#1}}						
	\renewcommand{\matrix}[1]{\bm{#1}}					
	
	\newcommand{\oSP}[1]{{#1}^\star}						
	\newcommand{\oRef}[2]{{#1}_{\mathrm{Ref}\addWithPreComma{#2}}}	
	\newcommand{\oErr}[1]{\tilde{#1}}					
	\newcommand{\oEq}[1]{\hat{#1}}						
	\newcommand{\oModify}[1]{\bar{#1}}					
	
	\newcommand{\Reals}{\ensuremath{\mathbb{R}}}		
	\newcommand{\RealsPos}{\ensuremath{\mathbb{R}_{\!\!\:>\!\!\;0}}}	
	\newcommand{\RealsPosDef}{\ensuremath{\mathbb{R}_{\!\!\:\ge\!\!\;0}}}	
	
	\newcommand{\classC}[1]{\ensuremath{C^{#1}}}		
	
	\newcommand{\vOneCol}[1][]{\vec{\mathds{1}}_{#1}}	
	\newcommand{\Ident}[1][]{\matrix{I}_{#1}}			
	
	\newcommand{\Transpose}{T}							
	
	\newcommand{\posDef}{\ensuremath{\succ}}			
	\newcommand{\posSemiDef}{\ensuremath{\succcurlyeq}}	
	\newcommand{\negSemiDef}{\ensuremath{\preccurlyeq}}	
	
	\newcommand{\setUnderBrace}[2]{\underset{\textstyle{}#1}{\underbrace{#2}}}
	\newcommand{\setUnderBraceMatrix}[2]{\,\,\underset{\textstyle{}#1}{\underbrace{\!\!#2\!\!}}\,\,}
	
	\DeclareMathOperator{\Diag}{diag}					
	
	\newlength\mytemplena
	\newlength\mytemplenb
	\DeclareDocumentCommand\myalignalign{sm}%
	{%
		\settowidth{\mytemplena}{$\displaystyle #2$}%
		\setlength\mytemplenb{\widthof{$\displaystyle=$}/2}%
		\hskip-\mytemplena%
		\hskip\IfBooleanTF#1{-\mytemplenb}{+\mytemplenb}%
	}

\newcommand{\varBase}[2]{%
	\expandafter\newcommand\csname s#1dot\endcsname[1][]{\dot{#2}_{##1}}%
	\expandafter\newcommand\csname v#1\endcsname[1][]{\vec{#2}_{##1}}%
	\expandafter\newcommand\csname v#1dot\endcsname[1][]{\dot{\vec{#2}}_{##1}}%
	\expandafter\newcommand\csname v#1T\endcsname[1][]{\vec{#2}_{##1}^\Transpose}%
}

\newcommand{\varBaseError}[2]{%
	\expandafter\newcommand\csname s#1e\endcsname[1][]{\oErr{#2}_{##1}}%
	\expandafter\newcommand\csname s#1edot\endcsname[1][]{\dot{\oErr{#2}}_{##1}}%
	\expandafter\newcommand\csname v#1e\endcsname[1][]{\oErr{\vec{#2}}_{##1}}%
	\expandafter\newcommand\csname v#1edot\endcsname[1][]{\dot{\oErr{\vec{#2}}}_{##1}}%
	\expandafter\newcommand\csname v#1eT\endcsname[1][]{\oErr{\vec{#2}}_{##1}^\Transpose}%
	\expandafter\newcommand\csname s#1eq\endcsname[1][]{\oEq{#2}_{##1}}%
	\expandafter\newcommand\csname s#1eqdot\endcsname[1][]{\dot{\oEq{#2}}_{##1}}%
	\expandafter\newcommand\csname v#1eq\endcsname[1][]{\oEq{\vec{#2}}_{##1}}%
}

\newcommand{\varModified}[3]{%
	\expandafter\newcommand\csname s#1#2dot\endcsname[1][]{\csname s#1dot\endcsname[#3\addWithPreComma{##1}]}%
	\expandafter\newcommand\csname v#1#2\endcsname[1][]{\csname v#1\endcsname[#3\addWithPreComma{##1}]}%
	\expandafter\newcommand\csname v#1#2dot\endcsname[1][]{\csname v#1dot\endcsname[#3\addWithPreComma{##1}]}%
	\expandafter\newcommand\csname v#1#2T\endcsname[1][]{\csname v#1T\endcsname[#3\addWithPreComma{##1}]}%
}

\newcommand{\varModifiedError}[3]{%
	\expandafter\newcommand\csname s#1#2e\endcsname[1][]{\csname s#1e\endcsname[#3\addWithPreComma{##1}]}%
	\expandafter\newcommand\csname s#1#2edot\endcsname[1][]{\csname s#1edot\endcsname[#3\addWithPreComma{##1}]}%
	\expandafter\newcommand\csname v#1#2e\endcsname[1][]{\csname v#1e\endcsname[#3\addWithPreComma{##1}]}%
	\expandafter\newcommand\csname v#1#2edot\endcsname[1][]{\csname v#1edot\endcsname[#3\addWithPreComma{##1}]}%
	\expandafter\newcommand\csname v#1#2eT\endcsname[1][]{\csname v#1eT\endcsname[#3\addWithPreComma{##1}]}%
	\expandafter\newcommand\csname s#1#2eq\endcsname[1][]{\csname s#1eq\endcsname[#3\addWithPreComma{##1}]}%
	\expandafter\newcommand\csname s#1#2eqdot\endcsname[1][]{\csname s#1eqdot\endcsname[#3\addWithPreComma{##1}]}%
	\expandafter\newcommand\csname v#1#2eq\endcsname[1][]{\csname v#1eq\endcsname[#3\addWithPreComma{##1}]}%
}

\newcommand{\varModifiedExp}[3]{%
	\expandafter\newcommand\csname s#1#2dot\endcsname[1][]{\csname s#1dot\endcsname[##1]^{#3}}%
	\expandafter\newcommand\csname v#1#2\endcsname[1][]{\csname v#1\endcsname[##1]^{#3}}%
	\expandafter\newcommand\csname v#1#2dot\endcsname[1][]{\csname v#1dot\endcsname[##1]^{#3}}%
	\expandafter\newcommand\csname v#1#2T\endcsname[1][]{\csname v#1T\endcsname[##1]^{#3}}%
}

\newcommand{\varModifiedExpError}[3]{%
	\expandafter\newcommand\csname s#1#2e\endcsname[1][]{\csname s#1e\endcsname[##1]^{#3}}%
	\expandafter\newcommand\csname s#1#2edot\endcsname[1][]{\csname s#1edot\endcsname[##1]^{#3}}%
	\expandafter\newcommand\csname v#1#2e\endcsname[1][]{\csname v#1e\endcsname[##1]^{#3}}%
	\expandafter\newcommand\csname v#1#2edot\endcsname[1][]{\csname v#1edot\endcsname[##1]^{#3}}%
	\expandafter\newcommand\csname v#1#2eT\endcsname[1][]{\csname v#1eT\endcsname[##1]^{#3}]}%
	\expandafter\newcommand\csname s#1#2eq\endcsname[1][]{\csname s#1eq\endcsname[##1]^{#3}}%
	\expandafter\newcommand\csname s#1#2eqdot\endcsname[1][]{\csname s#1eqdot\endcsname[##1]^{#3}}%
	\expandafter\newcommand\csname v#1#2eq\endcsname[1][]{\csname v#1eq\endcsname[##1]^{#3}}%
}

	\newcommand{\denoteSetNode}{\mathrm{s}}
	\newcommand{\denoteFollowNode}{\mathrm{f}}

	\newcommand{\denoteTx}{\mathrm{t}}
	\newcommand{\denoteVSC}{\mathrm{V}}
	\newcommand{\denoteLoad}{\mathrm{L}}
	\newcommand{\denotePassive}{\mathrm{P}}
	\newcommand{\denoteExternal}{\mathrm{ex}}
	\newcommand{\denoteInternal}{\mathrm{int}}
	\newcommand{\denotePort}{\mathrm{p}}
	\newcommand{\denoteFirst}{\mathrm{\alpha}}
	
	\newcommand{\denoteCapacitor}{\mathrm{C}}
	\newcommand{\denoteCluster}{\mathrm{c}}
	\newcommand{\denoteReduced}{\mathrm{r}}
	
	\newcommand{\denoteBus}{\mathrm{b}}
	\newcommand{\denoteNew}{\mathrm{n}}	
	\newcommand{\denoteFilterOut}{\mathrm{o}}
	
	\newcommand{\denoteSource}{\mathrm{in}}
	\newcommand{\denoteSink}{\mathrm{out}}
	
	\newcommand{\denoteImage}{\mathrm{I}}

	\newcommand{\denoteDGU}{\mathrm{d}}

	\newcommand{\denoteCrit}{\mathrm{crit}}
	\newcommand{\denoteEquivalent}{\mathrm{eq}}
	
\newcommand{\salpha}[1][]{\alpha_{#1}}
\newcommand{\salphaMod}[1][]{\oModify{\alpha}_{#1}}

\newcommand{\malpha}[1][]{\matrix{\alpha}_{#1}}
\newcommand{\malphaMod}[1][]{\oModify{\matrix{\alpha}}_{#1}}

\newcommand{\mA}[1][]{\matrix{A}_{#1}}
\newcommand{\mAT}[1][]{\matrix{A}_{#1}^\Transpose}

\newcommand{\mAs}[1][]{\mA[\denoteSetNode\addWithPreComma{#1}]}
\newcommand{\mAsT}[1][]{\mAT[\denoteSetNode\addWithPreComma{#1}]}

\newcommand{\mAr}[1][]{\mA[\denoteReduced\addWithPreComma{#1}]}
\newcommand{\mArT}[1][]{\mAT[\denoteReduced\addWithPreComma{#1}]}

\newcommand{\mAb}[1][]{\mA[\denoteBus\addWithPreComma{#1}]}
\newcommand{\mAbT}[1][]{\mAT[\denoteBus\addWithPreComma{#1}]}

\newcommand{\sbeta}[1][]{\beta_{#1}}
\newcommand{\sbetaMod}[1][]{\oModify{\beta}_{#1}}

\newcommand{\mbetaMod}[1][]{\oModify{\matrix{\beta}}_{#1}}

\newcommand{\vb}[1][]{\vec{b}_{#1}}
\newcommand{\vbT}[1][]{\vec{b}_{#1}^\Transpose}

\newcommand{\vbs}[1][]{\vb[\denoteSetNode\addWithPreComma{#1}]}
\newcommand{\vbsT}[1][]{\vbT[\denoteSetNode\addWithPreComma{#1}]}

\newcommand{\mB}[1][]{\matrix{B}_{#1}}
\newcommand{\mBT}[1][]{\matrix{B}_{#1}^\Transpose}

\newcommand{\mBr}[1][]{\mB[\denoteReduced\addWithPreComma{#1}]}
\newcommand{\mBrT}[1][]{\mBT[\denoteReduced\addWithPreComma{#1}]}

\newcommand{\sC}[1][]{C_{#1}}
\newcommand{\sCTx}[1][]{\sC[\denoteTx\addWithPreComma{#1}]}
\newcommand{\sCeq}[1][]{\sC[\denoteEquivalent\addWithPreComma{#1}]}

\newcommand{\mC}[1][]{\matrix{C}_{#1}}

\newcommand{\mCF}[1][]{\mC[\denoteFollowNode\addWithPreComma{#1}]}
\newcommand{\mCS}[1][]{\mC[\denoteSetNode\addWithPreComma{#1}]}

\newcommand{\sdelta}[1][]{\delta_{#1}}

\newcommand{\mdelta}[1][]{\matrix{\delta}_{#1}}
\newcommand{\mdeltaS}[1][]{\mdelta[\denoteSetNode\addWithPreComma{#1}]}
\newcommand{\mdeltaF}[1][]{\mdelta[\denoteFollowNode\addWithPreComma{#1}]}

\newcommand{\se}[1][]{e_{#1}}
\varBase{e}{e}

\newcommand{\mE}[1][]{\matrix{E}_{#1}}
\newcommand{\mET}[1][]{\matrix{E}_{#1}^\Transpose}

\newcommand{\mEF}[1][]{\mE[\denoteFollowNode\addWithPreComma{#1}]}
\newcommand{\mEFT}[1][]{\mET[\denoteFollowNode\addWithPreComma{#1}]}
\newcommand{\mES}[1][]{\mE[\denoteSetNode\addWithPreComma{#1}]}
\newcommand{\mEST}[1][]{\mET[\denoteSetNode\addWithPreComma{#1}]}

\newcommand{\mESetT}[1][]{\mE[\setT\addWithPreComma{#1}]}
\newcommand{\mESetTT}[1][]{\mET[\setT\addWithPreComma{#1}]}

\newcommand{\mEn}[1][]{\mE[\denoteNew\addWithPreComma{#1}]}
\newcommand{\mEnT}[1][]{\mET[\denoteNew\addWithPreComma{#1}]}

\newcommand{\setE}[1][]{\mathcal{E}_{#1}}
\newcommand{\setESource}[1][]{\setE[\denoteSource\addWithPreComma{#1}]}
\newcommand{\setESink}[1][]{\setE[\denoteSink\addWithPreComma{#1}]}

\newcommand{\vf}[1][]{\vec{f}_{#1}}

\newcommand{\setF}[1][]{\mathcal{F}_{#1}}

\newcommand{\graphG}[1][]{\mathcal{G}_{#1}}

\newcommand{\vh}[1][]{\vec{h}_{#1}}

\newcommand{\sH}[1][]{H_{#1}}
\newcommand{\sHdot}[1][]{\dot{H}_{#1}}

\newcommand{\sHL}[1][]{\sH[\denoteLoad\addWithPreComma{#1}]}

\newcommand{\sHr}[1][]{\sH[\denoteReduced\addWithPreComma{#1}]}
\newcommand{\sHrdot}[1][]{\sHdot[\denoteReduced\addWithPreComma{#1}]}

\newcommand{\sHs}[1][]{\sH[\denoteSetNode\addWithPreComma{#1}]}
\newcommand{\sHsdot}[1][]{\sHdot[\denoteSetNode\addWithPreComma{#1}]}

\newcommand{\sHTx}[1][]{\sH[\denoteTx\addWithPreComma{#1}]}
\newcommand{\sHTxdot}[1][]{\sHdot[\denoteTx\addWithPreComma{#1}]}

\newcommand{\sii}[1][]{i_{#1}}
\varBase{i}{i}
\varBaseError{i}{i}

\newcommand{\siC}[1][]{\sii[\denoteCapacitor\addWithPreComma{#1}]}
\varModified{i}{C}{\denoteCapacitor}

\newcommand{\sio}[1][]{\sii[\denoteFilterOut\addWithPreComma{#1}]}
\varModified{i}{o}{\denoteFilterOut}
\varModifiedError{i}{o}{\denoteFilterOut}

\newcommand{\siInt}[1][]{\sii[\denoteInternal\addWithPreComma{#1}]}
\varModified{i}{Int}{\denoteInternal}
\varModifiedError{i}{Int}{\denoteInternal}

\newcommand{\siEx}[1][]{\sii[\denoteExternal\addWithPreComma{#1}]}
\varModified{i}{Ex}{\denoteExternal}
\varModifiedError{i}{Ex}{\denoteExternal}

\varModified{i}{ExS}{\denoteExternal,\denoteSetNode}
\varModifiedError{i}{ExS}{\denoteExternal,\denoteSetNode}

\varModified{i}{ExF}{\denoteExternal,\denoteFollowNode}
\varModifiedError{i}{ExF}{\denoteExternal,\denoteFollowNode}

\varModified{i}{p}{\denotePort}
\varModifiedError{i}{p}{\denotePort}

\varModified{i}{pF}{\denotePort\denoteFollowNode}
\varModifiedError{i}{pF}{\denotePort\denoteFollowNode}

\varModified{i}{pS}{\denotePort\denoteSetNode}
\varModifiedError{i}{pS}{\denotePort\denoteSetNode}

\newcommand{\siTx}[1][]{\sii[\denoteTx\addWithPreComma{#1}]}
\varModified{i}{Tx}{\denoteTx}
\varModifiedError{i}{Tx}{\denoteTx}

\varModified{i}{SetT}{\setT}
\varModifiedError{i}{SetT}{\setT}

\varModified{i}{S}{\denoteSetNode}
\varModifiedError{i}{S}{\denoteSetNode}

\varModified{i}{Sa}{\denoteSetNode,\denoteFirst}
\varModifiedError{i}{Sa}{\denoteSetNode,\denoteFirst}

\varModified{i}{Fa}{\denoteFollowNode,\denoteFirst}
\varModifiedError{i}{Fa}{\denoteFollowNode,\denoteFirst}

\varModified{i}{F}{\denoteFollowNode}
\varModifiedError{i}{F}{\denoteFollowNode}

\newcommand{\sia}[1][]{\sii[\denoteFirst\addWithPreComma{#1}]}
\varModified{i}{a}{\denoteFirst}
\varModifiedError{i}{a}{\denoteFirst}

\newcommand{\sI}[1][]{I_{#1}}
\varBase{I}{I}
\varBaseError{I}{I}

\newcommand{\siL}[1][]{\sii[\denoteLoad\addWithPreComma{#1}]}
\varModified{i}{L}{\denoteLoad}
\varModifiedError{i}{L}{\denoteLoad}

\newcommand{\siLv}[1][]{\siL[#1](\sv[#1])}

\newcommand{\siLP}[1][]{\siL[#1]^{\denotePassive}}
\varModifiedExp{iL}{P}{\denotePassive}
\varModifiedExpError{iL}{P}{\denotePassive}

\newcommand{\siLPv}[1][]{\siLP[#1](\sv[#1])}
\newcommand{\siLPeve}[1][]{\siLPe[#1](\sve[#1])}
\newcommand{\siLPveq}[1][]{\siLP[#1](\sveq[#1])}

\newcommand{\IdentS}[1][]{\Ident[|\setS|]}
\newcommand{\IdentF}[1][]{\Ident[|\setF|]}

\newcommand{\skappa}[1][]{\kappa_{#1}}

\newcommand{\seig}[1][]{\lambda_{#1}}
\newcommand{\seigLapTwo}[1][]{\seig[2](\mLap[#1])}

\newcommand{\mEig}[1][]{\Lambda_{#1}}

\newcommand{\sL}[1][]{L_{#1}}
\newcommand{\sLTx}[1][]{\sL[\denoteTx\addWithPreComma{#1}]}

\newcommand{\mL}[1][]{\matrix{L}_{#1}}
\newcommand{\mLTx}[1][]{\mL[\denoteTx\addWithPreComma{#1}]}
\newcommand{\mLSetT}[1][]{\mL[\setT\addWithPreComma{#1}]}
\newcommand{\mLF}[1][]{\mL[\denoteFollowNode\addWithPreComma{#1}]}
\newcommand{\mLS}[1][]{\mL[\denoteSetNode\addWithPreComma{#1}]}

\newcommand{\mLap}[1][]{\matrix{\mathcal{L}}_{#1}}
\newcommand{\mLapT}[1][]{\mLap[#1]^\Transpose}

\newcommand{\setM}[1][]{\mathcal{M}_{#1}}
\newcommand{\setMF}[1][]{\setM[\denoteFollowNode]}
\newcommand{\setMS}[1][]{\setM[\denoteSetNode]}

\newcommand{\setN}[1][]{\mathcal{N}_{#1}}

\newcommand{\spsi}[1][]{\psi_{#1}}
\newcommand{\spsiMod}[1][]{\oModify{\psi}_{#1}}

\newcommand{\spsiV}[1][]{\spsi[\sv\addWithPreComma{#1}]}
\newcommand{\spsiVMod}[1][]{\spsiMod[\sv\addWithPreComma{#1}]}
\newcommand{\spsiI}[1][]{\spsi[\sii\addWithPreComma{#1}]}
\newcommand{\spsiIMod}[1][]{\spsiMod[\sii\addWithPreComma{#1}]}
\newcommand{\spsiInt}[1][]{}
\newcommand{\spsiIntSqr}[1][]{}

\newcommand{\mpsi}[1][]{\matrix{\psi}_{#1}}
\newcommand{\mpsiMod}[1][]{\oModify{\matrix{\psi}}_{#1}}

\newcommand{\mpsiVMod}[1][]{\mpsiMod[\sv\addWithPreComma{#1}]}

\newcommand{\mpsiIMod}[1][]{\mpsiMod[\sii\addWithPreComma{#1}]}
\newcommand{\mpsiInt}[1][]{\mpsi[\sInt\addWithPreComma{#1}]}

\newcommand{\sP}[1][]{P_{#1}}

\newcommand{\spp}[1][]{p_{#1}}

\newcommand{\sps}[1][]{\spp[#1]}

\newcommand{\mP}[1][]{\matrix{P}_{#1}}
\newcommand{\mPT}[1][]{\mP[#1]^\Transpose}

\newcommand{\mPs}[1][]{\mP[\denoteSetNode\addWithPreComma{#1}]}

\newcommand{\mPr}[1][]{\mP[\denoteReduced\addWithPreComma{#1}]}
\newcommand{\mPrT}[1][]{\mPr[#1]^\Transpose}

\newcommand{\mQ}[1][]{\matrix{Q}_{#1}}

\newcommand{\mQs}[1][]{\mQ[\denoteSetNode\addWithPreComma{#1}]}
\newcommand{\mQr}[1][]{\mQ[\denoteReduced\addWithPreComma{#1}]}

\newcommand{\siOut}[1][]{\rho_{#1}}

\newcommand{\sR}[1][]{R_{#1}}
\newcommand{\sRTx}[1][]{\sR[\denoteTx\addWithPreComma{#1}]}

\newcommand{\mR}[1][]{\matrix{R}_{#1}}
\newcommand{\mRinv}[1][]{\mR[#1]^{-1}}
\newcommand{\mRTx}[1][]{\mR[\denoteTx\addWithPreComma{#1}]}
\newcommand{\mRTxinv}[1][]{\mRinv[\denoteTx\addWithPreComma{#1}]}
\newcommand{\mRTxn}[1][]{\mR[\denoteTx,\denoteNew\addWithPreComma{#1}]}
\newcommand{\mRTxninv}[1][]{\mRinv[\denoteTx,\denoteNew\addWithPreComma{#1}]}
\newcommand{\mRSetT}[1][]{\mR[\setT\addWithPreComma{#1}]}
\newcommand{\mRF}[1][]{\mR[\denoteFollowNode\addWithPreComma{#1}]}
\newcommand{\mRS}[1][]{\mR[\denoteSetNode\addWithPreComma{#1}]}

\newcommand{\somega}[1][]{\omega_{#1}}
\newcommand{\somegaMod}[1][]{\oModify{\omega}_{#1}}

\newcommand{\somegaV}[1][]{\somega[\sv\addWithPreComma{#1}]}
\newcommand{\somegaVMod}[1][]{\somegaMod[\sv\addWithPreComma{#1}]}

\newcommand{\momegaMod}[1][]{\oModify{\matrix{\omega}}_{#1}}

\newcommand{\momegaVMod}[1][]{\momegaMod[\sv\addWithPreComma{#1}]}

\newcommand{\setS}[1][]{\mathcal{S}_{#1}}

\newcommand{\setT}[1][]{\mathcal{T}_{#1}}
\newcommand{\setTSource}[1][]{\setT[\denoteSource\addWithPreComma{#1}]}
\newcommand{\setTSink}[1][]{\setT[\denoteSink\addWithPreComma{#1}]}

\newcommand{\stau}[1][]{\tau_{#1}}
\newcommand{\mtau}[1][]{\matrix{\tau}_{#1}}

\newcommand{\su}[1][]{u_{#1}}
\varBase{u}{u}
\varBaseError{u}{u}

\newcommand{\sus}[1][]{\su[\denoteSetNode\addWithPreComma{#1}]}
\varModified{u}{s}{\denoteSetNode}
\varModifiedError{u}{s}{\denoteSetNode}

\varModified{u}{L}{\denoteLoad}
\varModifiedError{u}{L}{\denoteLoad}

\varModified{u}{up}{\denotePort}
\varModifiedError{u}{up}{\denotePort}

\varModified{u}{pS}{\denotePort\denoteSetNode}
\varModifiedError{u}{pS}{\denotePort\denoteSetNode}

\varModified{u}{pF}{\denotePort\denoteFollowNode}
\varModifiedError{u}{pF}{\denotePort\denoteFollowNode}

\varModified{u}{c}{\denoteCluster}
\varModifiedError{u}{c}{\denoteCluster}

\varModified{u}{SetT}{\setT}
\varModifiedError{u}{SetT}{\setT}

\newcommand{\sv}[1][]{v_{#1}}
\varBase{v}{v}
\varBaseError{v}{v}

\newcommand{\svRef}[1][]{\oRef{v}{#1}}
\newcommand{\svCrit}{\sv[\denoteCrit]}
\newcommand{\svSP}[1][]{\oSP{v}_{#1}}
\newcommand{\vvSP}[1][]{\oSP{\vec{v}}_{#1}}

\varModified{v}{S}{\denoteSetNode}
\varModifiedError{v}{S}{\denoteSetNode}

\newcommand{\vvSSP}[1][]{\vvSP[\denoteSetNode\addWithPreComma{#1}]}

\varModified{v}{F}{\denoteFollowNode}
\varModifiedError{v}{F}{\denoteFollowNode}

\newcommand{\svVSC}[1][]{\sv[\denoteVSC\addWithPreComma{#1}]}
\varModified{v}{VSC}{\denoteVSC}
\varModifiedError{v}{VSC}{\denoteVSC}

\varModified{v}{VSCF}{\denoteVSC,\denoteFollowNode}
\varModifiedError{v}{VSCF}{\denoteVSC,\denoteFollowNode}

\varModified{v}{VSCS}{\denoteVSC,\denoteSetNode}
\varModifiedError{v}{VSCS}{\denoteVSC,\denoteSetNode}

\newcommand{\mV}[1][]{\matrix{V}_{#1}}
\newcommand{\mVinv}[1][]{\mV[#1]^{-1}}
\newcommand{\mVI}[1][]{\mV[\denoteImage\addWithPreComma{#1}]}

\newcommand{\sw}[1][]{w_{#1}}

\newcommand{\vw}[1][]{\vec{w}_{#1}}
\newcommand{\vwT}[1][]{\vw[#1]^\Transpose}

\newcommand{\sx}[1][]{x_{#1}}
\varBase{x}{x}
\varBaseError{x}{x}
\newcommand{\sxSP}[1][]{\oSP{x}_{#1}}

\varModified{x}{c}{\denoteCluster}
\varModifiedError{x}{c}{\denoteCluster}

\varModified{x}{s}{\denoteSetNode}
\varModifiedError{x}{s}{\denoteSetNode}

\varModified{x}{r}{\denoteReduced}
\varModifiedError{x}{r}{\denoteReduced}

\newcommand{\setXeq}[1][]{\oEq{\mathcal{X}}_{#1}}

\newcommand{\sInt}[1][]{\xi_{#1}}
\varBase{Int}{\xi}
\varBaseError{Int}{\xi}

\varModified{Int}{S}{\denoteSetNode}
\varModifiedError{Int}{S}{\denoteSetNode}

\varModified{Int}{F}{\denoteFollowNode}
\varModifiedError{Int}{F}{\denoteFollowNode}

\varBase{y}{y}
\varBaseError{y}{y}

\varModified{y}{s}{\denoteSetNode}
\varModifiedError{y}{s}{\denoteSetNode}

\varModified{y}{L}{\denoteLoad}
\varModifiedError{y}{L}{\denoteLoad}

\varModified{y}{c}{\denoteCluster}
\varModifiedError{y}{c}{\denoteCluster}

\varModified{y}{SetT}{\setT}
\varModifiedError{y}{SetT}{\setT}

\newcommand{\sY}[1][]{Y_{#1}}
\newcommand{\mY}[1][]{\matrix{Y}_{#1}}
\newcommand{\mYF}[1][]{\mY[\denoteFollowNode\addWithPreComma{#1}]}
\newcommand{\mYS}[1][]{\mY[\denoteSetNode\addWithPreComma{#1}]}

\varBase{z}{z}
\varBaseError{z}{z}

\newcommand{\sZinv}[1][]{Z_{#1}^{-1}}

\newcommand{\sZCritinv}[1][]{\sZinv[\denoteCrit\addWithPreComma{#1}]}

\graphicspath{{./03_Img/}}

\def\BibTeX{{\rm B\kern-.05em{\sc i\kern-.025em b}\kern-.08em
		T\kern-.1667em\lower.7ex\hbox{E}\kern-.125emX}}

\begin{document}
	
	\title{Passivation of Clustered DC Microgrids with Non-Monotone Loads}
	\author{Albertus J.\ Malan, Joel Ferguson, Michele Cucuzzella, Jacquelien M.\ A.\ Scherpen, and S{\"o}ren Hohmann
		\thanks{
			(\emph{Corresponding author: A.\ J.\ Malan.})}
		\thanks{A.\ J.\ Malan, and S.\ Hohmann are with the Institute of Control Systems (IRS), Karlsruhe Institute of Technology (KIT), 76131, Karlsruhe, Germany. (E-mails: albertus.malan@kit.edu, soeren.hohmann@kit.edu).}
		\thanks{J.\ Ferguson is with the School of Engineering, The University of Newcastle, Australia (E-mail: joel.ferguson@newcastle.edu.au).}
		\thanks{M.\ Cucuzzella, and J.\ M.\ A.\ Scherpen are with the Faculty of Science and Engineering, University of Groningen, the Netherlands (E-mails: m.cucuzzella@rug.nl, j.m.a.scherpen@rug.nl).}
	}
	
	
	\maketitle
	
	\begin{abstract}

%
In this paper, we consider the problem of voltage stability in DC networks containing uncertain loads with non-monotone incremental impedances and where the steady-state power availability is restricted to a subset of the buses in the network.
We propose controllers for powered buses that guarantee voltage regulation and output strictly equilibrium independent passivity (OS-EIP) of the controlled buses, while buses without power are equipped with controllers that dampen their transient behaviour. The OS-EIP of a cluster containing both bus types is verified through a linear matrix inequality (LMI) condition, and the asymptotic stability of the overall microgrid with uncertain, non-monotone loads is ensured by interconnecting the OS-EIP clusters.
By further employing singular perturbation theory, we show that the OS-EIP property of the clusters is robust against certain network parameter and topology changes.


%

%

%

%

\end{abstract}

	
	\section{Introduction} \label{sec:Introduction}
\acresetall
%
%
%
\IEEEPARstart{C}{onstant} power loads present one of the most significant challenges in achieving stable voltage regulation in DC networks (also called microgrids \cite{Dragicevic2016,Meng2017}). Such load characteristics, where the power demand is independent of the voltage level at the load, are typically observed in power electronics or in complex electrical systems (see \cite{Singh2017,ALNussairi2017}). They also form an integral part of the widely used constant impedance (Z), constant current (I), and constant power (P) load models. If not properly managed, however, the negative incremental impedance exhibited by P loads can result in voltage oscillations or even voltage collapse in the network \cite{ALNussairi2017, Singh2017, Jusoh2004, Jeeninga2023}.
A further complication is presented by the often sparse availability of power in the DC network, which limits where in the network the bus voltages can be regulated.
Ensuring a stable network operation thus requires overcoming such negative incremental impedances which are only connected indirectly via dynamical transmission lines to the voltage regulating buses.
Motivated by this problem, we propose voltage controllers for buses with and without power available in the steady state which can ensure network stability, even for uncertain loads with negative incremental impedance located at arbitrary locations in the network.

\paragraph*{Literature Review}
Achieving a stable network operation by regulating the bus voltages has received significant attention in the literature (see \cite{Dragicevic2016,Meng2017,AlIsmail2021,Modu2023} and the sources therein). These approaches include droop controllers \cite{Zhao2015}; plug-and-play capable controllers \cite{Tucci2016,Sadabadi2018,Tucci2018stab}; energy-based controllers \cite{Kosaraju2021,Cucuzzella2023}; passivity based controllers \cite{Nahata2020,Strehle2020dc}; controllers employing linearisation \cite{Perez2020, Machado2021}; and general integral-augmented state-feedback controllers \cite{Ferguson2021,Silani2022,Ferguson2023}.

Nevertheless, many publications only consider loads comprising Z and/or I components \cite{Tucci2016,Tucci2018stab,Sadabadi2018,Perez2020,Kosaraju2021}; or restrict the load parameters to ensure that the incremental load impedances are monotone increasing functions \cite{Nahata2020,Strehle2020dc,Ferguson2021}. More recent results in \cite{Machado2021,Cucuzzella2023,Ferguson2023} allow loads with non-monotone incremental impedances. In \cite{Machado2021}, voltage regulation is achieved through an input-output linearisation, but an estimate of the load parameters are required. In \cite{Cucuzzella2023}, a Krasovskii-type storage function is used to design the controller which requires a possibly noisy voltage derivative. In \cite{Ferguson2023}, stability and an \ac{ISS} property proven analytically for non-monotone ZIP loads.

These approaches typically only ensure stability by restricting the network. This includes permitting only a single bus \cite{Perez2020,Machado2021}; assuming there is power available for the voltage regulation at each bus \cite{Tucci2016,Sadabadi2018,Nahata2020,Strehle2020dc,Ferguson2021,Ferguson2023}; or by assuming a Kron-reduced network \cite{Tucci2018stab,Silani2022}. Yet while the Kron reduction (see \cite{Doerfler2013}) allows Z and I loads to be shifted to buses where the voltage can be regulated, this procedure can fail for nonlinear loads or loads with negative parameters \cite{Chen2021}. The remaining cases which allow loads at buses without voltage regulation require such loads to be have monotone incremental impedances \cite{Kosaraju2021}.

\paragraph*{Main Contributions}
In this paper, we aim to achieve voltage stability in a DC microgrid containing buses with and without power available at steady state and containing uncertain loads with non-monotone incremental impedances.
Inspired by \cite{Cucuzzella2023}, where the bus voltage derivative is used to regulate a bus with non-monotone loads, we propose a similar controller that passivates a bus where steady-state power is available. We improve upon \cite{Cucuzzella2023}, however, by replacing the voltage derivative with the easily measured current of the filter capacitor. Furthermore, we show that buses with non-monotone loads and where no steady-state power is available cannot be passivated. Instead, we design a controller which dampens the transients caused by non-monotone loads at such buses. After combining the closed-loop buses with and without steady-state power into a cluster, we verify the passivity of the entire cluster.

The contributions in this paper thus comprise:
\begin{enumerate}
	\item A voltage setting controller for buses where steady-state power is available along with analytical conditions ensuring \ac{OSEIP} for these buses with uncertain, non-monotone loads.
	\item A voltage following controller for buses without steady-state power, where the controller dampens the bus transients without passivating the bus.
	\item A \ac{LMI} which can be used to verify the \ac{OSEIP} of a cluster containing buses with and without steady-state power.
	\item The asymptotic stability of a DC microgrid comprising an arbitrary number of interconnected \ac{OSEIP} clusters.
	\item A reduced-order \ac{LMI} to verify the cluster \ac{OSEIP} obtained by applying singular perturbation theory.
	\item An investigation into the robustness of the cluster passivity against parameter or topology changes.
\end{enumerate}

\paragraph*{Paper Organisation}
The introduction concludes with some notation and preliminaries. In \autoref{sec:Problem}, we introduce the microgrid model and formulate objectives for the control problem using the model notation. Thereafter, we provide the bus controllers in \autoref{sec:Control} and investigate the properties of the buses. \autoref{sec:Clustering} then follows with an investigation of the cluster \ac{OSEIP} and microgrid stability. The robustness of the results are investigated in \autoref{sec:Analysis} and demonstrated via simulation in \autoref{sec:Simulation}. Concluding remarks are provided in \autoref{sec:Conclusion}.

\paragraph*{Notation and Preliminaries}
Define as a vector $\vec{a} = (a_k)$ and a matrix $\matrix{A} = (a_{kl})$. The vector $\vOneCol[k]$ is a $k$-dimensional vector of ones and $\Ident[k]$ is the identity matrix of dimension $k$.
The determinant of a matrix $\matrix{A}$ is $\det(\matrix{A})$ and $\Diag[\cdot]$ creates a (block-)diagonal matrix from the supplied vector (or matrices).
Let $\mA \posDef 0$ ($\posSemiDef 0$) denote a symmetric positive (semi-)definite matrix.
The set of positive real numbers is defined by $\RealsPos$.
For a variable $\sx$, we denote its unknown steady state as $\sxeq$, its error state as $\sxe \coloneqq \sx - \sxeq$, and a desired setpoint as $\sxSP$. Furthermore, $\sx \equiv 0$ indicates that $\sx$ and all its derivates are zero for all $t \ge 0$.
Whenever clear from context, we omit the time dependence of variables.


Consider a nonlinear system
\begin{equation} \label{eq:Prelim:NL_System}
	\left\{\begin{aligned}
		\vxdot &= \vf(\vx,\vu), \\
		\vy &= \vh(\vx),
	\end{aligned}\right.
\end{equation}
where $\vx \in \Reals^{n}$, $\vu \in \Reals^{p}$, and $\vy \in \Reals^{p}$, and where $\vf\colon\Reals^{n}\times\Reals^{p}\rightarrow\Reals^{n}$ and $\vh\colon\Reals^{n}\rightarrow\Reals^{p}$ are class \classC{1} functions.
\begin{definition}[Dissipative system, See \cite{vdSchaft2017,Arcak2006}] \label{def:Prelim:disspativity}
	A system \eqref{eq:Prelim:NL_System} with a class \classC{1} storage function $\sH \colon \Reals^n \times \Reals^{p} \rightarrow \RealsPosDef$, $\sH(\vec{0}) = 0$, is dissipative w.r.t.\ a supply rate $\sw(\vu,\vy)$ if $\sHdot \le \sw(\vu,\vy)$. 
\end{definition}
%
%
%
\begin{definition}[{\Ac{ZSO}, \Ac{ZSD} \cite[p.~46f]{vdSchaft2017}}] \label{def:Passive:ZSO}
	A system \eqref{eq:Prelim:NL_System} is \ac{ZSO} if $\vu \equiv \vec{0}$ and $\vy \equiv \vec{0}$ implies $\vx \equiv \vec{0}$ and \ac{ZSD} if $\vu \equiv \vec{0}$ and $\vy \equiv \vec{0}$ implies $\lim_{t\to\infty} \vx = \vec{0}$.
\end{definition}

%
For cases where the desired equilibrium of a system is not at the origin but at some possibly unknown constant value, the shifted passivity \cite[p.~96]{vdSchaft2017} or \ac{EIP} \cite{Hines2011} of a system is investigated. This requires that an equilibrium exists, i.e.\ there is a unique input $\vueq \in \Reals^m$ for every equilibrium $\vxeq \in \setXeq \subset \Reals^n$ such that \eqref{eq:Prelim:NL_System} verifies $\vf(\vxeq,\vueq) = \vec{0}$ and $\vyeq = \vh(\vxeq,\vueq)$ \cite[p.~24]{Arcak2016}.
\begin{definition}[{\Ac{EIP}}] \label{def:Passive:EIP}
	A system \eqref{eq:Prelim:NL_System} with a class \classC{1} storage function $\sH(\vx, \vxeq)$, $\cramped{\sH \colon \Reals^n \times \setXeq \rightarrow \RealsPosDef}$, with $\sH(\vxeq, \vxeq) = 0$, that is dissipative w.r.t.\ $\vueT\vye - \siOut\vyeT\vye$ for all $\cramped{(\vx, \vxeq, \vu)} \in \Reals^n \times \setXeq \times \Reals^p$, is \ac{EIP} if $\siOut \ge 0$ and \ac{OSEIP} if $\siOut > 0$.
%
\end{definition}
%
%

	\section{Modelling and Problem Description} \label{sec:Problem}
\begin{figure}[!t]
	\centering
	\resizebox{\columnwidth}{!}{\input{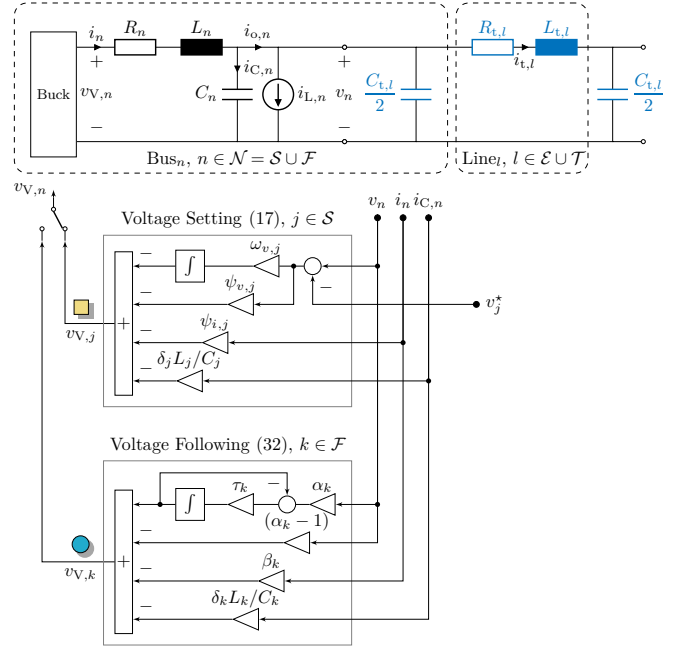}}
	\caption{A bus comprising a DC-DC buck converter, an LC filter and a nonlinear load, which connects to a $\pi$-model transmission line ({\color{matlabCol1}blue}). Each bus is equipped with a voltage setting controller ($j \in \setS$) or a voltage following controller ($k \in \setF$).}
	\label{fig:Problem:DC_Elements}
\end{figure}
\begin{figure}[!t]
	\centering
	\resizebox{0.75\columnwidth}{!}{\begin{tikzpicture}
	\def\nodeXdist{1.5cm}
	\def\nodeYdist{1.15cm}
	\def\nodeXshift{\nodeXdist*0.75}
	\def\boundR{0.65cm}
	
	\def\Nodes{1,2,3,4,5,6,7,8,9,10,11,12,13,14,15,16,17,18,19,20,21}
	\def\NodesVSet{3/left,8/above left,14/above,21/right}
	\def\NodesVFollow{1/below,2/left,4/below,5/below,6/left,7/above,9/left,10/left,11/above,12/above left,13/above right,15/above,16/above,17/above,18/right,19/above,20/right}
	
	\def\InternalEdges{1/2, 2/3, 3/4, 3/5, 6/7, 7/8, 8/9, 9/10, 8/11, 10/12, 11/12, 13/14, 14/15, 14/16, 15/16, 17/18, 17/20, 18/19, 18/21, 19/20, 20/21}
	\def\ExternalEdges{1/7/-45/225, 4/15/30/180, 8/13/60/300, 12/20/0/240, 16/18/0/120} 
	
	\tikzstyle{nodeVSet}     		= [draw, rectangle, minimum width = 8pt, minimum height = 8pt, drop shadow={shadow scale=1.1,fill=black!35}, fill = buff]
	\tikzstyle{nodeVsetSmall}		= [draw, rectangle, minimum width = 4pt, minimum height = 4pt, drop shadow={shadow scale=1.1,fill=black!35}, fill = buff]
	\tikzstyle{nodeVFollow}  		= [draw, circle, drop shadow={shadow scale=1.1,fill=black!35}, radius=0.2pt, fill = ballblue]
	\tikzstyle{nodeVFollowSmall}	= [draw, circle, drop shadow={shadow scale=1.1,fill=black!35}, radius=0.1pt, fill = ballblue]
	
	
	\tikzstyle{lineInternal}		= [-, decoration={markings, mark=at position 0.55 with {\arrow{Triangle[scale=0.85]}}}, postaction={decorate}]
	\tikzstyle{lineExternal}		= [-,thick, double, decoration={markings, mark=at position 0.55 with {\arrow{Triangle[scale=0.6]}}}, postaction={decorate}]
	\tikzstyle{backgroundCluster}	= [fill=black!10, draw, dashed]

	
	\coordinate(cNode1);
	\path (cNode1) ++(-\nodeXshift,\nodeYdist*0.5) coordinate (cNode2) ++(0,\nodeYdist) coordinate (cNode3) ++(\nodeXshift,\nodeYdist*0.5) coordinate (cNode4) ++(0,-\nodeYdist) coordinate (cNode5);
	
	\path (cNode1) ++(\nodeXdist*1.0,-\nodeYdist*0.8) coordinate (cNode7) ++(\nodeXshift*1.1,\nodeYdist*0.7) coordinate (cNode8) ++(\nodeXshift*1.1,-\nodeYdist/2) coordinate (cNode11) ++(0,-\nodeYdist) coordinate (cNode12);
	\path (cNode7) ++(0,-\nodeYdist) coordinate (cNode6);
	\path (cNode8) ++(0,-\nodeYdist) coordinate (cNode9) ++(0,-\nodeYdist) coordinate (cNode10);
	
	\path (cNode8) ++(0,\nodeYdist*1.6) coordinate (cNode13) ++(0,\nodeYdist*0.9) coordinate (cNode14) ++(-\nodeXshift,\nodeYdist*0.8) coordinate (cNode15);
	\path (cNode14) ++(\nodeXshift,\nodeYdist*0.8) coordinate (cNode16);
	
	\path (cNode13) ++(\nodeXdist*1.3,-\nodeYdist*0.7) coordinate (cNode17) ++(\nodeXdist*0.9,0) coordinate (cNode21) ++(\nodeXdist*0.9,0) coordinate (cNode19);
	\path (cNode21) ++(0,\nodeYdist) coordinate (cNode18);
	\path (cNode21) ++(0,-\nodeYdist) coordinate (cNode20);
	\path (cNode17) ++(\nodeXdist*0.7,0) coordinate (cNode21);
	
	\foreach \a/\b in \NodesVSet {
		\node[nodeVSet](node\a) at(cNode\a) {};
		\node[\b=4pt](nodeText\a) at(cNode\a) {\a};
	}
	\foreach \a/\b in \NodesVFollow {
		\node[nodeVFollow](node\a) at(cNode\a) {};
		\node[\b=4pt](nodeText\a) at(cNode\a) {\a};
	}
	
	\foreach \a/\b in \InternalEdges {
		\draw[lineInternal](node\a) -- (node\b);	}
	\foreach \a/\b/\c/\d in \ExternalEdges {
		\draw[lineExternal](node\a) -- (node\b);	}

	\begin{scope}[on background layer]
		
		\path[backgroundCluster] \convexpath{node1,node2,node3,node4}{\boundR};
		\path[backgroundCluster] \convexpath{node6,node7,node8,node11,node12,node10}{\boundR};
		\path[backgroundCluster] \convexpath{node13,node15,node16}{\boundR};
		\path[backgroundCluster] \convexpath{node17,node18,node19,node20}{\boundR};
		
		\path (node3) -- (node4) node[above=0.95cm,midway,font=\large] {Cluster 1};
		\path (node10) -- (node10) node[below=0.7cm,pos=0.0,font=\large] {Cluster 2};
		\path (node15) -- (node16) node[above=0.7cm,midway,font=\large] {Cluster 3};
		\path (node18) -- (node19) node[above=0.75cm,pos=0.0,font=\large] {Cluster 4};

	\end{scope}
\end{tikzpicture}}
	\caption[A clustered DC microgrid.]{A 21-bus DC microgrid partitioned into four clusters. Each cluster comprises buses with voltage setting controllers 
		\begin{tikzpicture}
			\protect\node[draw, rectangle, inner sep=-2, minimum size=1.5ex, fill = buff] {};
		\end{tikzpicture}
		in $\setS$, buses with voltage following controllers 
		\begin{tikzpicture}
			\protect\node[draw, circle, inner sep=-2, minimum size=1.5ex, fill = ballblue] {};
		\end{tikzpicture}
		in $\setF$, and lines with arbitrary directions
		\begin{tikzpicture}[baseline=-0.75ex]
			\tikzstyle{lineInternal}	= [-, decoration={markings, mark=at position 0.6 with {\arrow{Triangle[scale=0.9]}}}, postaction={decorate}]
			\protect\draw[lineInternal] (0,0) -- (0.6cm,0);
		\end{tikzpicture}
		in $\setE$. Lines
		\begin{tikzpicture}[baseline=-0.75ex]
			\tikzstyle{lineInternal}	= [-, double, decoration={markings, mark=at position 0.7 with {\arrow{Triangle[scale=0.9]}}}, postaction={decorate}]
			\protect\draw[lineInternal] (0,0) -- (0.6cm,0);
		\end{tikzpicture}
		in $\setT$ interconnect buses in different clusters.
	}
	\label{fig:Problem:Example_network}
\end{figure}
In this paper, we consider buses comprising a nonlinear load together with a buck inverter connected via a lossy LC filter, as depicted at the top of \autoref{fig:Problem:DC_Elements}. This configuration is typically used to model \acp{DGU}, but can similarly be used to model other grid-connected devices with local energy storage, such as electric vehicles and smart loads. The buses are interconnected with a broader network of similar devices via transmission lines. Furthermore, the buses in the same arbitrarily defined neighbourhood collectively form a cluster, as shown in \autoref{fig:Problem:Example_network}. The entire microgrid can then be interpreted as the interconnection of several clusters, each of which comprises one or more buses.

As each bus could be either a \ac{DGU} or another device with local storage, it is important to note that the available energy at each bus may be varying. \acp{DGU} are considered to have an unbounded energy supply for control purposes, whereas a device with local storage will have only a finite energy supply to be utilised for control. This motivates the need to distinguish two classes of buses: voltage setting buses that have unlimited local energy supply, and voltage following buses that have only a finite local energy supply. Two separate control objectives are posed for each type of bus. For voltage setting buses, the control objective is \emph{voltage setpoint regulation}, i.e., the buck converter is controlled such that the bus converges to a pre-determined voltage in steady state. Conversely, for the voltage following buses, the control objective is \emph{transient stabilisation}. Thus, the buck converter is controlled to ensure stability of the bus, but the steady-state voltage of the bus depends the rest of the network. Importantly, in the latter case there is no power supply in steady state. Recalling the cluster structure of the network in \autoref{fig:Problem:Example_network}, it is assumed that each cluster contains at least one supply bus.

More formally, we consider a microgrid of weakly connected buses described by the graph $\graphG[\mathrm{MG}] = (\bigcup_{m \in \setM} \setN[m], \, \setT \,\cup\, \bigcup_{m \in \setM} \setE[m])$, where the buses and lines are partitioned into a set of clusters $\cramped{\setM = \lbrace1,\dots,M\rbrace}$, $\cramped{M \ge 1}$. The sets $\setN[m]$ and $\setE[m]$ denotes the buses and lines associated with a cluster $m$, whereas $\setT$ are the lines not associated with any cluster, e.g., lines connecting buses in different clusters. Note that each bus appears in exactly one cluster, i.e., $\setN[m_1] \cap \setN[m_2] = \emptyset$ for any $m_1, m_2 \in \setM$ with $m_1 \ne m_2$. Transmission lines connecting buses in the same cluster are assigned to $\setE[m]$ such that the buses $\setN[m]$ each cluster $m$ are weakly connected.

The rest of this section continues as follows: in \autoref{sec:Problem:Elec_Network}, we model the cluster along with its constituent components, and in \autoref{sec:Problem:Control_Problem}, we formulate the control problem to be addressed in the sequel.
\subsection{DC Cluster Model} \label{sec:Problem:Elec_Network}
Let a cluster in a DC microgrid be represented as a graph $\graphG[m] = (\setN[m], \setE[m])$ comprising $|\setN[m]|$ buses connected by $|\setE[m]|$ electrical lines, with $\setE[m] \subseteq \setN[m] \times \setN[m]$. Each bus consists of a nonlinear load, a DC buck converter and a lossy LC filter, and each line is represented using the well-known $\pi$-model \cite[p.~66]{Machowski2008}, as depicted in \autoref{fig:Problem:DC_Elements}.
The buses in the cluster are partitioned into two sets. The first set, indicated by $\setS[m] = \lbrace1,\dots,S_m\rbrace, S_m \ge 1$, has unlimited power supply and aims to achieve local voltage regulation. The second set, indicated by $\setF[m] = \lbrace S_m+1,\dots,S_m+F_m\rbrace$, $F_m \ge 0$, has limited power for control purposes and aims to achieve transient stabilisation via a \emph{voltage following controller}. Importantly, steady-state power for regulation or stabilisation is only available at the buses in $\setS[m]$.

In the following, we introduce models for the transmission lines, the loads and for the voltage setting and voltage following buses. These components are subsequently combined into the state-space model for the entire DC cluster. For simplicity, we omit the cluster index $m$ for the rest of this section.
\subsubsection{Line Model}
Each transmission line is assigned an arbitrary direction, which denotes the positive current flow.
The $\pi$-model transmission line dynamics interconnecting a source bus $n_{l,\denoteSource}$ and a sink bus $n_{l,\denoteSink}$ are described by
\begin{equation} \label{eq:Problem:line_dynamics}
	\sLTx[l]\siTxdot[l] = -\sRTx[l]\siTx[l] + \sv[n_{l,\denoteSource}] - \sv[n_{l,\denoteSink}] , \quad l \in \setE \cup \setT,
\end{equation}
where $\siTx[l] \in \Reals$ is the line current, $\sLTx[l], \sRTx[l] > 0$ are the line inductance and resistance, and $\sv[n_{l,\denoteSource}], \sv[n_{l,\denoteSink}] \in \RealsPosDef$ are the corresponding bus voltages.
Note that the $\pi$-model capacitances are included in the dynamics of the respective source and sink buses.
\subsubsection{Load Model}
The load at each bus $n \in \setN$ is modelled as a static, nonlinear voltage-dependent current source described by a class \classC{0} function, which we describe using the standard ZIP-model in this work. As per \cite[pp.~110--112]{Machowski2008}, we define a critical voltage $\svCrit > 0$, typically set to $\SI{70}{\percent}$ of the reference voltage $\svRef$, below which the loads are purely resistive. Then, the load current is described by
\begin{align} \label{eq:Problem:ZIP_model}
	\siLv[n] =& \left\{\begin{array}{lll}
		\sZinv[n]\sv[n] + \sI[n] + \dfrac{\sP[n]}{\sv[n]}, & \qquad & \sv[n] \ge \svCrit, \\[1pt]
		\sZCritinv[n]\sv[n], & \qquad & \sv[n] < \svCrit ,
	\end{array} \right. \\[2pt]
	\label{eq:Problem:Z_crit}
	\sZCritinv[n] \coloneqq& \; \frac{\siL[n](\svCrit)}{\svCrit} = \sZinv[n] + \frac{\sI[n]}{\svCrit} + \frac{\sP[n]}{\svCrit^2} \, ,
\end{align}
where $\sZinv[n], \sI[n], \sP[n] \in \Reals$ describe the constant load impedance, current and power components, respectively, and $\sv[n] \in \RealsPosDef$ is the bus voltage.

\begin{figure}[!t]
	\centering
	\resizebox{0.85\columnwidth}{!}{\begin{tikzpicture}[>=latex']
    
    \def\pZ{0.05}
    \def\pI{-0.25}
    \def\pP{0.83102}
    
    \def\pY{0}
    
    \def\cVnom{1}
    \def\cVmax{2.2}
    \def\cVcrit{0.7}

    \def\Xlim{\cVmax}
    \def\Ylim{2}
    \def\domainX{-1:\Xlim}
    \def\rangeY{0:\Ylim}
    \def\scaleXY{1.8}
    \def\scaleY{0.5}
    
    \def\domainXa{0:\cVcrit-0.15}
    \def\domainXb{\cVcrit-0.15:\cVmax}
    \def\domainXFull{0:\cVmax}
    
    
    \begin{scope}[range=-\rangeY, scale=\scaleXY]
    \end{scope}

    \begin{scope}[range=\rangeY, scale=\scaleXY, thick]
		\draw[] plot[id=vcrit] coordinates{(\cVcrit,-0.05) (\cVcrit,0.05)};
		\node[anchor=south] at (\cVcrit,0) {$\svCrit$};
	
    	\draw[domain=\domainXa, color=matlabCol1, line cap = round] plot[id=load_low] function{\scaleY* (\pZ+\pI/\cVcrit+\pP/\cVcrit/\cVcrit)*x};
    	\draw[domain=\domainXb, color=matlabCol1, line cap = round] plot[id=load_full,samples=150] function{\scaleY*((\pZ+\pI/\cVcrit+\pP/\cVcrit/\cVcrit)*x * (x < \cVcrit) + (\pZ*x+\pI+\pP/x) * (x >= \cVcrit))}
    		node[right=0.15, anchor = west, rectangle, fill=white!50, inner sep=0.2] {\textcolor{black}{$\textcolor{matlabCol1}{\siLv[n]} \eqqcolon \textcolor{matlabCol2}{\siLPv[n]} - \textcolor{matlabCol4}{\sY[n] \sv[n]}$}};
    	
    	\draw[domain=\domainXa, color=matlabCol2, line cap = round] plot[id=load_pas_low] function{\scaleY* (\pZ+\pI/\cVcrit+\pP/\cVcrit/\cVcrit)*x - \scaleY*(\pZ-\pP/\cVcrit/\cVcrit)*x};
    	\draw[domain=\domainXb, color=matlabCol2, line cap = round] plot[id=load_pas_full,samples=150] function{\scaleY*((\pZ+\pI/\cVcrit+\pP/\cVcrit/\cVcrit)*x * (x < \cVcrit) + (\pZ*x+\pI+\pP/x) * (x >= \cVcrit)) - \scaleY*(\pZ-\pP/\cVcrit/\cVcrit)*x} 
    		node[right=0.15, anchor = west, rectangle, fill=white!50, inner sep=0.2] {$\siLPv[n]$};
    	
    	\draw[domain=\domainXFull, color=matlabCol4, line cap = round] plot[id=load_lin] function{-\scaleY*(\pZ-\pP/\cVcrit/\cVcrit)*x}
    		node[right=0.15, anchor = north west, rectangle, fill=white!50, inner sep=0.2] {$\sY[n] \sv[n]$};
    		
    \end{scope}

	\draw[-,thick, scale=\scaleXY, line cap = round] (0,\Ylim) -- (0,0) coordinate[pos=0.5] (cMidY) -- (\Xlim+0.2,0) node[below left=0.25, circle, inner sep=-0.4, align=center] {} coordinate[pos=0.5] (cMidX);

	\node[anchor=north] at(cMidX) {Voltage};
	\node[rotate=90,anchor=south] at(cMidY) {Current};
\end{tikzpicture}}
	\caption{A nonlinear ZIP load ({\color{matlabCol1}blue}) which is separated into a linear function ({\color{matlabCol4}purple}) and a monotone increasing nonlinear function ({\color{matlabCol2}orange}).}
	\label{fig:Problem:Split_Load}
\end{figure}
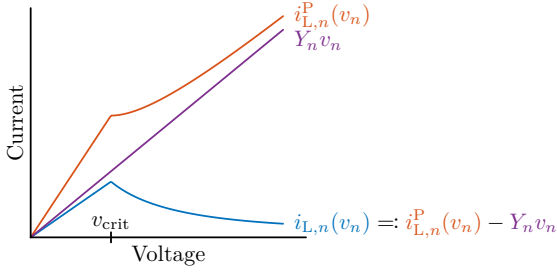
A typical current response of a ZIP load is shown in \autoref{fig:Problem:Split_Load} by the blue curve. The nonlinear load curve often exhibits a lack of monotonicity, which can be detrimental to the system stability. Assuming that the load is globally Lipschitz, we therefore separate the load into a linear part with gradient $\sY[n] \in \Reals$ and a monotone increasing nonlinear part $\siLPv[n]$ according to
%
\begin{equation} \label{eq:Problem:load_eip}
	\siLv[n] \eqqcolon \siLPv[n] - \sY[n] \sv[n],
\end{equation}
where the gradient $\sY[n]$ is chosen as small as possible while ensuring that $\siLPv[n]$ is monotone. The separation in \eqref{eq:Problem:load_eip} is illustrated in \autoref{fig:Problem:Split_Load}. Note that a suitable value for $\sY[n]$ is load-dependent. This parameter is used in the sequel to determine control gains for both the voltage setting and voltage following controllers.

\begin{remark} \label{rem:Problem:load_Y_value}
	In many applications, the exact load $\siL[n]$ is unknown, making the selection of a suitable $\sY[n]$ difficult. In such cases, choosing a larger $\sY[n]$ results in a higher robustness against load uncertainty, but will also lead to larger control gains. 
\end{remark}
\begin{remark} \label{rem:Problem:load_IFP}
	The smallest $\sY[n]$ yielding a monotone $\siLP[n]$ corresponds to the \ac{IFP} index of the load (cf.\ \cite[Proposition 9]{Malan2023}), with $\sY[n] > 0$ indicating a lack of \ac{IFP}. If $\siL[n]$ is already monotone, $\sY[n] = 0$ may be chosen, whereas $\sY[n] > 0$ is required if $\siL[n]$ is not monotone.
\end{remark}
\begin{remark} \label{rem:Problem:Nonlinear_Loads}
	Beyond the ZIP loads we consider here, the results in the sequel also apply to any continuous nonlinear static load which is globally Lipschitz and can be expressed using \eqref{eq:Problem:load_eip}. This includes exponential load models (see e.g.\ \cite[p.~112]{Machowski2008}, \cite{Strehle2020Load}).
\end{remark}
\subsubsection{Voltage Setting Buses}
Consider now a bus $j \in \setS \subseteq \setN$ where power is available to regulate the voltage. The bus comprises a buck converter, a lossy LC filter and the load in \eqref{eq:Problem:ZIP_model}, as shown in \autoref{fig:Problem:DC_Elements}. Let $\setESource[j] \subseteq \setE$ and $\setTSource[j] \subseteq \setT$ be the sets of edges of which bus $j$ is the sink, and similarly $\setESink[j] \subseteq \setE$ and $\setTSink[j] \subseteq \setT$ the sets of edges of which bus $j$ is the source.
The bus dynamics can be described using the time-averaged converter model and an ideal energy source (see e.g.\ \cite{Trip2019, Strehle2020dc,Tucci2016}) as
\begingroup
\allowdisplaybreaks
\begin{align}
	\label{eq:Problem:bus_S_dynamics_passive_load}
	\begin{bmatrix}
		\sL[j]\sidot[j] \\ \sCeq[j]\svdot[j]
	\end{bmatrix} \!\!&=\! \begin{bmatrix}
		-\sR[j] & -1 \\
		1 & \sY[j]
	\end{bmatrix} \!\! \begin{bmatrix}
		\sii[j] \\ \sv[j]
	\end{bmatrix} \! + \! \begin{bmatrix}
		\svVSC[j] \\ \siInt[j] + \siEx[j]
	\end{bmatrix} \!, \!\quad j \in \setS,
	\\
	\label{eq:Problem:bus_S_internal_current}
	\siInt[j] & \coloneqq \sum_{l \in \setESource[j]} \siTx[l] - \sum_{l \in \setESink[j]} \siTx[l] ,
	\\
	\label{eq:Problem:bus_S_external_current}
	\siEx[j] & \coloneqq \sum_{l \in \setTSource[j]} \siTx[l] - \sum_{l \in \setTSink[j]} \siTx[l] - \siLPv[j] ,
\end{align}
\endgroup
where $\sL[j], \sR[j] > 0$ are the filter inductance and resistance, respectively, and $\svVSC[j]$ is the control input. The line capacitances are described by $\sCTx[l] > 0$, $l \in \setESource[j] \cup \setESink[j] \cup \setTSource[j] \cup \setTSink[j]$ and are combined with the filter capacitance $\sC[j] > 0$ to obtain the total effective capacitance seen by the bus
\begin{equation} \label{eq:Problem:Equivalent_Cap}
	\sCeq[j] = \sC[j] + \sum_{l \in \setESource[j] \cup \setESink[j] \cup \setTSource[j] \cup \setTSink[j]} \tfrac{1}{2}\sCTx[l].
\end{equation}
%
The term $\siInt[j]$ is the sum of all line currents into and out of bus $j$ from transmission lines within the cluster $\setE$. Furthermore, the term $\siEx[j]$ forms an exogenous input to the bus dynamics which consists of the line currents interconnecting the bus $j$ to buses in other clusters $\setT$ and the monotone load part $\siLPv[j]$, defined in \eqref{eq:Problem:load_eip}.
\subsubsection{Voltage Following Buses} \label{sec:Problem:Elec_Network:S_Bus}
Next, consider a bus $k \in \setF \subset \setN$ where no steady-state power is available for regulation or stabilisation. As shown in \autoref{fig:Problem:DC_Elements}, such a bus consists of the same components as the voltage setting buses. However, since voltage setting and voltage following buses aim to achieve different objectives, different controllers will be designed for them. The uncontrolled dynamics of a voltage following bus are identical to \eqref{eq:Problem:bus_S_dynamics_passive_load}, i.e.,
\begin{align}
	\label{eq:Problem:bus_F_dynamics_passive_load}
	\begin{bmatrix}
		\sL[k]\sidot[k] \\ \sCeq[k]\svdot[k]
	\end{bmatrix} \!\!&=\! \begin{bmatrix}
		-\sR[k] & -1 \\
		1 & \sY[k]
	\end{bmatrix} \!\! \begin{bmatrix}
		\sii[k] \\ \sv[k]
	\end{bmatrix} \!\! + \!\! \begin{bmatrix}
		\svVSC[k] \\ \siInt[k] \!+\! \siEx[k]
	\end{bmatrix} \!\!, \!\!\!\!\quad k \in \setF,
	\\
	\label{eq:Problem:bus_F_internal_current}
	\siInt[k] & \coloneqq \sum_{l \in \setESource[k]} \siTx[l] - \sum_{l \in \setESink[k]} \siTx[l] ,
	\\
	\label{eq:Problem:bus_F_external_current}
	\siEx[k] & \coloneqq \sum_{l \in \setTSource[k]} \siTx[l] - \sum_{l \in \setTSink[k]} \siTx[l] - \siLPv[k] ,
\end{align}
where the variable and parameter definitions are the same as before.
\subsubsection{Cluster Model}
By combining the dynamics of the voltage setting buses \eqref{eq:Problem:bus_S_dynamics_passive_load}, the voltage following buses \eqref{eq:Problem:bus_F_dynamics_passive_load}, and the lines \eqref{eq:Problem:line_dynamics}, we get the state-space model for a cluster in vector form
\begin{equation} \label{eq:Problem:cluster_dynamics}
	\scalebox{0.95}{$
		\begin{bmatrix}
			\mLS \viSdot \\ \mCS \vvSdot \\ \mLF \viFdot \\ \mCF \vvFdot \\ \mLTx \viTxdot
		\end{bmatrix} \! {=} \! \begin{bmatrix}
			-\mRS & -\IdentS & \matrix{0} & \matrix{0} & \matrix{0} \\
			\IdentS & \mYS & \matrix{0} & \matrix{0} & \mES \\
			\matrix{0} & \matrix{0} & -\mRF & -\IdentF & \matrix{0} \\
			\matrix{0} & \matrix{0} & \IdentF & \mYF & \mEF \\
			\matrix{0} & -\mEST & \matrix{0} & -\mEFT & -\mRTx
		\end{bmatrix} \!\! \begin{bmatrix}
			\viS \\ \vvS \\ \viF \\ \vvF \\ \viTx
		\end{bmatrix}  
		\!\!+\!\! \begin{bmatrix}
			\vvVSCS \\ \viExS \\ \vvVSCF \\ \viExF \\ \vec{0}
		\end{bmatrix} \!,
	$}
\end{equation}
where $\viS, \vvS, \vvVSCS, \viExS$, and $\viF, \vvF, \vvVSCF, \viExF$ denote the stacked vectors of variables for the buses in $\setS$ and $\setF$, respectively, while $\viTx$ denotes the stacked current vectors for the edges in $\setE$. Furthermore, $\mLS = \Diag[(\sL[j])]$, $\mCS = \Diag[(\sCeq[j])]$, $\mLF = \Diag[(\sL[k])]$, $\mCF = \Diag[(\sCeq[k])]$, $\mLTx = \Diag[(\sLTx[l])]$, $\mRTx = \Diag[(\sRTx[l])]$ for $j \in \setS$, $k \in \setF$, and $l \in \setE$.
The interconnections between buses and edges are described by the incidence matrix of the graph $\graphG = (\setN, \setE)$, i.e.,
\begin{equation} \label{eq:Problem:interconnection}
	\mE = (\se[nl]) \in \lbrace-1,0,1\rbrace^{|\setN|\times|\setE|}, \quad \mE = \begin{bmatrix}
		\mES \\ \mEF
	\end{bmatrix} ,
\end{equation}
where $\se[nl] = 1$ if bus $n$ is the sink of edge $l$, $\se[nl] = -1$ if bus $n$ is the source of edge $l$, and $\se[nl] = 0$ otherwise. Note that $\mE$ has been partitioned into $\mES \in \lbrace-1,0,1\rbrace^{|\setS|\times|\setE|}$ which describes the interconnections of the voltage setting nodes and $\mEF \in \lbrace-1,0,1\rbrace^{|\setF|\times|\setE|}$ which describes the interconnections of the voltage following nodes. Using the incidence matrix \eqref{eq:Problem:interconnection}, the internal currents \eqref{eq:Problem:bus_S_internal_current} and \eqref{eq:Problem:bus_F_internal_current} can be included directly in the state matrix in \eqref{eq:Problem:cluster_dynamics}. Note that the set of lines in $\setT$ interconnecting buses in different clusters are not included in the description in \eqref{eq:Problem:cluster_dynamics}.

The cluster dynamics \eqref{eq:Problem:cluster_dynamics} will be controlled to a non-zero equilibrium in the sequel. Specifically, the voltage setting buses will control the local voltages $\vvS$ to some user-defined $\vvSSP$, whereas the voltage following buses will control the local current $\viF$ to zero. Due to the nonlinear nature of the loads, we assume the existence of a corresponding equilibrium point for the cluster.
\begin{assumption} \label{ass:Problem:Cluster_equilibrium}
	Given a constant input $(\vvVSCS, \viExS, \vvVSCF, \viExF) = (\vvVSCSeq, \viExSeq, \vvVSCFeq, \viExFeq)$, the cluster dynamics \eqref{eq:Problem:cluster_dynamics} has an equilibrium at some $(\viS, \vvS, \viF, \vvF, \viTx) = (\viSeq, \vvSSP, \vec{0}, \vvFeq, \viTxeq)$ for $\viSeq \in \Reals^{|\setE|}$, $\vvFeq \in \RealsPosDef^{|\setF|}$ and $\viTxeq \in \Reals^{|\setT|}$.
\end{assumption}
\begin{remark} \label{rem:Problem:Equilibrium_uniqueness}
	Although the linear dynamics in \eqref{eq:Problem:cluster_dynamics} might suggest that the uniqueness of the equilibria can easily be investigated, the uniqueness is also dependent on the nonlinear parts of the static loads in $\viExS$ and $\viExF$.
\end{remark}
Now, by virtue of Assumption \ref{ass:Problem:Cluster_equilibrium}, the cluster dynamics \eqref{eq:Problem:cluster_dynamics} are shifted w.r.t.\ the equilibrium. Since the dynamics in \eqref{eq:Problem:cluster_dynamics} are linear, we use the error variables, e.g., $\viSe \coloneqq \viS - \viSeq$, to represent the shifted system, i.e., 
\begin{equation} \label{eq:Problem:cluster_dynamics_shifted}
	\scalebox{0.95}{$
		\begin{bmatrix}
			\mLS \viSedot \\ \mCS \vvSedot \\ \mLF \viFedot \\ \mCF \vvFedot \\ \mLTx \viTxedot
		\end{bmatrix} \! {=} \! \begin{bmatrix}
			-\mRS & -\IdentS & \matrix{0} & \matrix{0} & \matrix{0} \\
			\IdentS & \mYS & \matrix{0} & \matrix{0} & \mES \\
			\matrix{0} & \matrix{0} & -\mRF & -\IdentF & \matrix{0} \\
			\matrix{0} & \matrix{0} & \IdentF & \mYF & \mEF \\
			\matrix{0} & -\mEST & \matrix{0} & -\mEFT & -\mRTx
		\end{bmatrix} \!\! \begin{bmatrix}
			\viSe \\ \vvSe \\ \viFe \\ \vvFe \\ \viTxe
		\end{bmatrix}  
		\!\!+\!\! \begin{bmatrix}
			\vvVSCSe \\ \viExSe \\ \vvVSCFe \\ \viExFe \\ \vec{0}
		\end{bmatrix} \!.
		$}
\end{equation}
Note that the nonlinear parts of the loads $\siLPv[n]$, $n \in \setN$ are fully contained in the exogenous input variables $\viExS$ and $\viExF$. Shifting these input variables yields
\begin{multline} \label{eq:Problem:shifted_passive_load}
	\siLPeve[n] \coloneqq \, \siLPv[n] - \siLPveq[n] \\ = \siLP[n](\sve[n] + \sveq[n]) - \siLPveq[n], \quad n \in \setN .
\end{multline}
\begin{remark} \label{rem:Problem:External_Load_EIP}
	From \eqref{eq:Problem:shifted_passive_load} we see that $\siLPe[n](0) = 0$. Moreover, since $\siLPv[n]$ is a monotone function (see \autoref{fig:Problem:Split_Load}), the shifted function $\siLPeve[n]$ is also monotone. Thus, $\siLPeve[n]$ represents a static \ac{EIP} function with a storage function $\sHL[n] = 0$ (see \autoref{def:Passive:EIP} and \cite[p.~228ff.]{Khalil2002}).
\end{remark}
\subsection{Control Problem} \label{sec:Problem:Control_Problem}
The control problem considered in this paper is to design control laws for the voltage setting and voltage following buses, such that a cluster of such buses is \ac{OSEIP} with respect to some power port. Recalling \autoref{fig:Problem:Example_network}, if each cluster is passivated with respect to its interconnection port, the entire DC microgrid is stabilised by the inherited passivity properties. 
\begin{objective} \label{obj:Problem:Cluster_Passivity}
	Design decentralised controllers for each $\svVSC[n]$, $n \in \setS \cup \setF = \setN$ such that the cluster \eqref{eq:Problem:cluster_dynamics_shifted} is \ac{OSEIP} w.r.t.\ the input-output pairs $(\viExSe, \vvSe)$ and $(\viExFe, \vvFe)$. Moreover, the voltage setting and voltage following buses should satisfy the following sub-objectives:
    \begin{enumerate}
        \item The decentralised controller for each $\svVSC[j]$, $j \in \setS$ must ensure that $\lim_{t\to\infty} \sv[j] = \svSP[j]$ and that the bus is \ac{OSEIP} w.r.t.\ the input-output pair $(\siExe[j], \sve[j])$.
        \item The decentralised controller for each $\svVSC[k]$, $k \in \setF$ must ensure that $\lim_{t\to\infty} \sii[k] = 0 \iff \lim_{t\to\infty} (\sv[k] - \svVSC[k]) = 0 $.
    \end{enumerate}
\end{objective}
Once each cluster is passivated via the design of suitable control laws, the interconnection of multiple clusters will be stable. It is conceivable, however, that the topology of a cluster may change dynamically as buses are plugged in or out. Therefore, in addition to the above control objective, we introduce the following objective in order to characterise the robustness properties of the proposed control scheme with respect to possible network reconfigurations.
\begin{objective} \label{obj:Problem:Cluster_Robustness}
	Formulate conditions under which the cluster passivity in \autoref{obj:Problem:Cluster_Passivity} is preserved in the event of parameter or topology changes.
\end{objective}

	\section{Bus Controllers} \label{sec:Control}
In this section, we design the controllers outlined in \autoref{obj:Problem:Cluster_Passivity} for buses in $\setS$ where steady-state power is available, and for buses in $\setF$ without steady-state power. To this end, we start in \autoref{sec:Control:Voltage_Setting} with controllers for the voltage setting buses in $\setS$. We prove that the closed-loop voltage setting buses exhibit the desired voltage regulation equilibrium and we derive conditions for the \ac{OSEIP} of these buses. Next, in \autoref{sec:Control:Voltage_Following}, we show that buses with non-monotone loads and without steady-state power cannot be \ac{EIP}. We therefore design a controller that dampens the transients of the voltage following buses in $\setF$. 

Note that the microgrid index $m$ and the bus indices $j$ and $k$ are omitted in this section for clarity. Nevertheless, the control parameters can be chosen differently for each bus.
\subsection{Voltage Setting Controller} \label{sec:Control:Voltage_Setting}
For the \acp{DGU} in the set $\setS$ that can supply power to the network in steady state, we propose the following controller to regulate the bus voltage to a user-defined setpoint $\svSP \in \RealsPosDef$,
\begin{subequations} \label{eq:Control:Voltage_setting}
	\begin{align}
		\sIntdot &= \somegaV(\sv - \svSP), \\
		\svVSC &= -\spsiV(\sv - \svSP) - \spsiI \sii - \sInt - \tfrac{\sdelta \sL}{\sC} \siC,
	\end{align}
\end{subequations}
where $\somegaV > 0$ and $\spsiV > 0$ are the integral and proportional control gains, $\spsiI > 0$ injects damping onto the dynamics of the filter current $\sii$ and $\sdelta > 0$ injects damping onto the dynamics of the capacitor current $\siC$.

We now construct a suitable error system for the bus dynamics \eqref{eq:Problem:bus_S_dynamics_passive_load} in closed-loop with \eqref{eq:Control:Voltage_setting}. For this development, we propose a new coordinate, i.e.,
\begin{equation} \label{eq:Control:split_currents}
	\sia \coloneqq \sii + \sdelta\sv, 
\end{equation}
which is useful for verifying a passivity property of the closed-loop bus dynamics. In particular, the control gain $\sdelta$ is used to inject damping onto the voltage dynamics.
\begin{proposition} \label{prop:Control:Closed_loop_voltage_setting}
	Consider the closed-loop comprising the bus dynamics \eqref{eq:Problem:bus_S_dynamics_passive_load} and the voltage setting controller \eqref{eq:Control:Voltage_setting}. Applying the coordinate transformation  in \eqref{eq:Control:split_currents} and shifting the system w.r.t.\ $\sInteq$, $\siaeq$, $\sveq = \svSP$, $\siInteq$, and $\siExeq$ yields 
	\begin{align}
		\label{eq:Control:Shifted_voltage_setting_system}
		\!\!\!\setUnderBraceMatrix{\mQs\vxsedot}{\begin{bmatrix}
			\sL\sIntedot \\ \sL\siaedot \\ \sCeq\svedot
		\end{bmatrix}} \! {=}& \! \setUnderBraceMatrix{\mAs}{\begin{bmatrix}
			0 & 0 & \somegaVMod \\
			-1\spsiInt & -\spsiIMod & -\spsiVMod \\
			0 & 1 & \sY - \sdelta
		\end{bmatrix}} \!\! \setUnderBraceMatrix{\vxse}{\begin{bmatrix}
			\sInte \\ \siae \\ \sve
		\end{bmatrix}} \!\:\! + \!\!\: \setUnderBraceMatrix{\vbs \suse}{\begin{bmatrix}
			0 \\ 0 \\ \siInte + \siExe
		\end{bmatrix}}, \\
		\label{eq:Control:Shifted_voltage_setting_output}
		\syse =& \setUnderBraceMatrix{\vbsT}{\,\begin{bmatrix}
				0 & 0 & 1
		\end{bmatrix}\,} \vxse = \sve, \\
		\label{eq:Control:modified_gain_psi_omega}
		\spsiIMod \coloneqq& \, \spsiI {+} \sR, \;\, \spsiVMod \coloneqq \spsiV + 1 - \sdelta\spsiIMod, \;\, \somegaVMod \coloneqq \sL\somegaV.
	\end{align}
\end{proposition}
\begin{proof}
	Starting from time derivative of \eqref{eq:Control:split_currents} weighted by $\sL$, use $\sL\sidot$ from \eqref{eq:Problem:bus_S_dynamics_passive_load} along with the capacitor equation $\sC \svdot = \siC$ to obtain
	\begin{equation} \label{eq:Control:Split_current_dynamics_voltage_setting_current}
		\sL\siadot = \sL \sidot + \sL \sdelta \svdot = -\sR\sii - \sv + \svVSC + \tfrac{\sL \sdelta}{\sC} \siC .
	\end{equation}
	Replacing $\svVSC$ with \eqref{eq:Control:Voltage_setting} and using $\sii = \sia - \sdelta\sv$, we obtain
	\begin{equation} \label{eq:Control:Split_current_dynamics_voltage_setting_current_2}
		\begin{aligned}
			\sL\siadot &= -(\sR + \spsiI)\sii - \sv -\spsiV(\sv - \svSP) - \spsiInt \sInt  \\
			&= -\spsiIMod \sia -\spsiVMod\sv - \spsiInt \sInt + \spsiV\svSP,
		\end{aligned}
	\end{equation}
	with $\spsiIMod$ and $\spsiVMod$ defined in \eqref{eq:Control:modified_gain_psi_omega}.
	Similarly, using \eqref{eq:Control:split_currents} in the voltage dynamics \eqref{eq:Problem:bus_S_dynamics_passive_load} yields
	\begin{equation} \label{eq:Control:Split_current_dynamics_voltage_setting_voltage}
		\sCeq \svdot = \sii + \sY + \siInt + \siEx = \sia + (\sY-\sdelta)\sv + \siInt + \siEx .
	\end{equation}
	Furthermore, multiplying $\sIntdot$ by $\sL$ yields
	\begin{equation} \label{eq:Control:Split_current_dynamics_voltage_setting_int}
		\sL\sIntdot = \somegaVMod(\sv - \svSP),
	\end{equation}
	with $\somegaVMod$ defined in \eqref{eq:Control:modified_gain_psi_omega}.
	To conclude, shifting the linear system \eqref{eq:Control:Split_current_dynamics_voltage_setting_current_2}--\eqref{eq:Control:Split_current_dynamics_voltage_setting_int} w.r.t.\ $\sInteq$, $\siaeq$, $\sveq = \svSP$, $\siInteq$, and $\siExeq$ yields the system \eqref{eq:Control:Shifted_voltage_setting_system}.
\end{proof}
Note that the transformation in \eqref{eq:Control:split_currents} changes neither the measurements used for control nor the shifted input-output port $(\siInte+\siExe, \sve)$.
We now investigate the \ac{OSEIP} for the controlled voltage setting bus \eqref{eq:Control:Shifted_voltage_setting_system}.
\begin{theorem} \label{thm:Control:V_setting_passive}
	Consider the controlled voltage setting bus in \eqref{eq:Control:Shifted_voltage_setting_system} and \eqref{eq:Control:Shifted_voltage_setting_output}. If there is a $\siOut > 0$ such that
	\begin{align}
		\label{eq:Control:V_set_Condition_1}
		\sdelta &> \sY + \siOut, \\
		\label{eq:Control:V_set_Condition_2}
		\spsiIMod^2 (\sdelta - \sY - \siOut) &> \somegaVMod \spsiInt - \spsiVMod\spsiIMod,
	\end{align}
	then there exists a $\sps[2] > 0$ 
	%
	%
	which ensures that the controlled voltage setting bus is \ac{OSEIP} w.r.t.\ the input-output pair $(\suse, \syse) = (\siInte+\siExe, \sve)$ with the storage function 
	\begin{equation}
		\label{eq:Control:V_set_Storage}
		\sHs(\vxse) = \tfrac{1}{2} \vxseT \mQs \mPs \vxse, \quad \mQs = \Diag[(\sL, \sL, \sCeq)], 
	\end{equation}
	where
	\begin{align}
		\label{eq:Control:V_set_p_matrix}
		\mPs &= \begin{bmatrix}
			\sps[1] & \sps[2] & 0 \\
			\sps[2] & \sps[3] & 0 \\
			0 & 0 & 1
		\end{bmatrix}, \quad \left\lbrace
		\begin{aligned}
			\sps[1] &\coloneqq \frac{1}{\spsiIMod}\sps[2] + \frac{1}{\spsiIMod\somegaVMod}, \\
			\sps[3] &\coloneqq \spsiIMod\sps[2].
		\end{aligned}\right.
	\end{align}
\end{theorem}
The proof of \autoref{thm:Control:V_setting_passive} is given in \autorefapp{sec:App:Proof_V_setting_passive}.
Through the conditions in \autoref{thm:Control:V_setting_passive}, the \ac{OSEIP} of the controlled voltage setting bus w.r.t.\ the input currents $\siInte+\siExe$ and the bus voltage error $\sve$ is ensured. Specifically, the condition in \eqref{eq:Control:V_set_Condition_1} ensures the \ac{OSEIP} if the damping injected using the capacitor current $\siC$ is strictly greater than the $\sY$ associated with the load (see \autoref{fig:Problem:Split_Load}). Building on the \ac{OSEIP} of the controlled voltage setting bus we now investigate its zero state.
\begin{proposition} \label{prop:Control:Voltage_Setting_zero_dynamics}
	The controlled voltage setting bus \eqref{eq:Control:Shifted_voltage_setting_system} with $\suse \equiv 0$ and $\syse \equiv 0$ is \ac{ZSO}.
\end{proposition}
\begin{proof}
	Setting $\suse \equiv \siInte + \siExe \equiv 0$ and $\syse \equiv \sve \equiv 0$ in \eqref{eq:Control:Shifted_voltage_setting_system} yields the reduced dynamics $\sCeq \svedot \equiv 0 \equiv \siae$ and thus that $\sL \siaedot \equiv 0 \equiv - \sInte$. 
	This implies $(\sInte, \siae, \sve)^\Transpose = \vxse = \vec{0}$.
\end{proof}
Using \autoref{prop:Control:Voltage_Setting_zero_dynamics} together with \autoref{thm:Control:V_setting_passive}, the asymptotic stability of the controlled voltage setting bus is assured if it is connected to another \ac{OSEIP} system on its $(\siInte+\siExe, \sve)$ port.
\begin{remark} \label{rem:Control_V_setting_robust}
	Even though the \ac{OSEIP} conditions in \autoref{thm:Control:V_setting_passive} are dependent on the bus parameters $\sR$, $\sL$ and the load parameter $\sY$, robustness against parameter changes can be ensured by verifying \eqref{eq:Control:V_set_Condition_1} and \eqref{eq:Control:V_set_Condition_2} for a range of $\sR$, $\sL$ and for the upper bound of $\sY$.
\end{remark}
%
%
\begin{remark} \label{rem:Control:Practical_derivative}
	Observe that the filter capacitor current $\siC$ and voltage derivative satisfy (see \autoref{fig:Problem:DC_Elements})
	\begin{equation} \label{eq:Control:damping_equivalence}
		\sC \svdot = \siC = \sii - \sio.
	\end{equation}
	Thus, including $\siC$ in the controller \eqref{eq:Control:Voltage_setting} has the same effect as injecting damping via the voltage derivative (see \cite{Cucuzzella2023}). Furthermore, note that instead of measuring the filter capacitor current $\siC$, one can measure the inductor current $\sii$ and the filter output $\sio$.
\end{remark}
\subsection{Voltage Following Controller} \label{sec:Control:Voltage_Following}
In the case where no steady-state power is available for control and where a bus has a non-monotone load, it is not possible to achieve an \ac{EIP} property on the $(\siInte+\siExe, \sve)$ port, as we demonstrate in the following proposition.
\begin{proposition} \label{prop:Control:no_power_no_eip}
	Consider a voltage bus \eqref{eq:Problem:bus_F_dynamics_passive_load} with no power available in steady state, i.e., $\lim_{t\to\infty} \sii = 0$, and where the load function $\siLv$ is not monotone. Such a bus cannot be \ac{EIP} w.r.t.\ to the input-output pair $(\siInte+\siExe, \sve)$.
\end{proposition}
\begin{proof}
	For a system to be passive, both its transient and steady-state behaviour must be passive. Considering the steady state $\sveqdot = 0$ of \eqref{eq:Problem:bus_F_dynamics_passive_load} along with the restricted power availability $\sieq = 0$, it follows that 
	\begin{equation} \label{eq:Control:no_power_steady_state}
		\sCeq\sveqdot = 0 = \sY\sveq + \siInteq + \siExeq 
		\implies -\sY\sveq = \siInteq + \siExeq.
	\end{equation}
	For the steady-state $(\siInteq+\siExeq, \sveq)$ relation in \eqref{eq:Control:no_power_steady_state} to be \ac{EIP}, $\sY \le 0$ is required \cite[Def.~6.1]{Khalil2002}, which means that the load must be monotone (see \autoref{rem:Problem:load_IFP}).
\end{proof}
There is thus no controller subject to $\lim_{t\to\infty} \sii = 0$ which can ensure the desired \ac{EIP} property for the buses in $\setF$ in the presence of non-monotone loads. Instead of passivating the bus, we therefore propose a voltage following controller which dampens the transient response of a bus voltage when a disturbance occurs, i.e.,
\begin{subequations} \label{eq:Control:Voltage_following}
	\begin{align}
		\sIntdot &= \stau(\salpha\sv - \sInt) , \\
		\svVSC &= -(\salpha-1)\sv - \sbeta \sii + \sInt - \tfrac{\sdelta \sL}{\sC} \siC ,
	\end{align}
\end{subequations}
where $\salpha$, $\sbeta$, $\sdelta \in \RealsPos$ are control gains chosen to inject damping onto the bus voltage $\sv$, the filter current $\sii$, and the capacitor current $\siC$, respectively. Furthermore, dynamics of $\sIntdot$ represent a low-pass filter where the time constant $\stau \in \RealsPos$ determines how quickly $\sInt$ follows the scaled bus voltage $\salpha\sv$ (see \autoref{fig:Problem:DC_Elements}). Note that \autoref{rem:Control:Practical_derivative} applies to the voltage following controller \eqref{eq:Control:Voltage_following} as well.

Similar to the voltage setting controller in \eqref{eq:Control:Voltage_setting}, the proportional damping factors $\salpha$, $\sbeta$, and $\sdelta$ allow the voltage following controller to react quickly to disturbances. Unlike the voltage setting controller, however, no voltage setpoint is provided in \eqref{eq:Control:Voltage_following}. Instead, the steady-state control output $\svVSCeq$ should be set such that \autoref{obj:Problem:Cluster_Passivity}.2 is achieved, i.e., $\sieq = 0$. This is attained via the integrator state $\sInt$, which follows the scaled bus voltage, resulting in $\svVSC \to \sv$ as $\sii \to 0$.

\begin{figure}[!t]
	\centering
	\begin{subfigure}[b]{0.38\columnwidth}
		\centering
		\includegraphics[width=\textwidth]{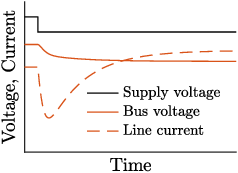}
		\caption{States vs.\ time}
	\end{subfigure}
	\begin{subfigure}[b]{0.5\columnwidth}
		\centering
		\includegraphics[width=\textwidth]{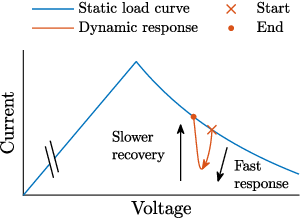}
		\caption{Voltage vs.\ current}
	\end{subfigure}
	\caption{Simulated line current and bus voltage for a transmission line connecting a bus with a static P load and a voltage following controller \eqref{eq:Control:Voltage_following} to a ideal voltage supply, (a) shows the trajectories over time and (b) plots the bus voltage vs. the line current.}
	\label{fig:Problem:Load_Curve}
\end{figure}
The control action is demonstrated with the setup in \autoref{fig:Problem:DC_Elements} by equipping the bus \eqref{eq:Problem:bus_F_dynamics_passive_load} with a voltage following controller \eqref{eq:Control:Voltage_following} and connecting the bus via a transmission line to an ideal voltage supply. The state trajectories resulting from a step change to the supply voltage are shown in \autoref{fig:Problem:Load_Curve}, where the bus voltage and the transmission line current are plotted against time and against each other. In steady state, the current delivered by the transmission line corresponds to the static load curve at the bus, since $\sieq = 0$ for the voltage following controller. After the supply voltage set, the line current quickly deviates from the load curve, indicating the stabilising control action from the proportional damping in \eqref{eq:Control:Voltage_following}. Thereafter, the line current increases and approaches the static load curve as $\sInt$ changes to ensure $\svVSC \to \sv$ and $\sii \to 0$.

Recalling \autoref{ass:Problem:Cluster_equilibrium}, we now apply the voltage following controller \eqref{eq:Control:Voltage_following} to the bus dynamics \eqref{eq:Problem:bus_F_dynamics_passive_load} and show that \autoref{obj:Problem:Cluster_Passivity}.2 is achieved.
\begin{proposition} \label{prop:Control:No_power_following}
	A voltage following bus \eqref{eq:Problem:bus_F_dynamics_passive_load} controlled with \eqref{eq:Control:Voltage_following} injects no power in steady state, i.e., $\sieq = 0$. Furthermore, applying the coordinate transform in \eqref{eq:Control:split_currents} and shifting the system w.r.t.\ $\sInteq = \salpha \sveq$, $\siaeq$, $\sveq$, $\siInteq$ and $\siExeq$, with $\sieq = \siaeq - \sdelta\sveq = 0$, yields
	\begin{align}
		\label{eq:Control:Shifted_no_power_system}
		\begin{bmatrix}
			\sIntedot \\ \sL\siaedot \\ \sCeq\svedot
		\end{bmatrix} \! {=}& \! \begin{bmatrix}
			-\stau & 0 & \stau \salpha \\
			1 & -\sbetaMod & -\salphaMod \\
			0 & 1 & \sY - \sdelta
		\end{bmatrix} \!\! \begin{bmatrix}
			\sInte \\ \siae \\ \sve
		\end{bmatrix} \!\:\! + \!\!\: \begin{bmatrix}
			0 \\ 0 \\ \siInte + \siExe
		\end{bmatrix} , \\
		\label{eq:Control:modified_gain_alpha_beta}
		\sbetaMod \coloneqq& \, \sR + \sbeta, \quad
		\salphaMod \coloneqq \,\salpha - \sbetaMod\sdelta .
	\end{align}
\end{proposition}
\begin{proof}
	Applying \eqref{eq:Control:Voltage_following} to the bus \eqref{eq:Problem:bus_F_dynamics_passive_load} yields
	\begin{equation} \label{eq:Control:controlled_following_bus}
		\begin{bmatrix}
			\sIntdot \\ \sL\sidot \\ \sCeq\svdot
		\end{bmatrix} \! = \! \begin{bmatrix}
			-\stau & 0 & \stau \salpha \\
			1 & -\sbetaMod & -\salpha \\
			0 & 1 & \sY
		\end{bmatrix} \!\! \begin{bmatrix}
			\sInt \\ \sii \\ \sv
		\end{bmatrix} \!\:\! + \!\!\: \begin{bmatrix}
			0 \\ -\tfrac{\sdelta \sL}{\sC} \siC \\ \siInt + \siEx
		\end{bmatrix} .
	\end{equation}
	Observe from \eqref{eq:Control:controlled_following_bus} in steady state that $\cramped{\sInteq = \salpha\sveq}$. Applying this to the second state in \eqref{eq:Control:controlled_following_bus} yields
	\begin{equation}
		\sL\sieqdot = 0 = \sInteq - \sbetaMod\sieq - \salpha\sveq - \sdelta\sL\sveqdot \iff 0 = - \sbetaMod\sieq ,
	\end{equation}
	since $\sveqdot = 0$. Thus, the steady state of \eqref{eq:Control:controlled_following_bus} satisfies $\sieq = 0$. Finally, similarly to the proof of \autoref{prop:Control:Closed_loop_voltage_setting}, by applying the coordinate transformation in \eqref{eq:Control:split_currents} to \eqref{eq:Control:controlled_following_bus} and shifting the system w.r.t.\ the equilibrium, we obtain the shifted dynamics in \eqref{eq:Control:Shifted_no_power_system}.
\end{proof}

	\section{Interconnected Passive Clusters} \label{sec:Clustering}
In order to achieve \autoref{obj:Problem:Cluster_Passivity}, we now use the results in \autoref{prop:Control:Closed_loop_voltage_setting} and \autoref{prop:Control:No_power_following} and analyse the passivity of an entire cluster. Whereas the voltage following buses are \ac{OSEIP} (see \autoref{thm:Control:V_setting_passive}) and the transmission lines \eqref{eq:Problem:line_dynamics} are similarly \ac{OSEIP} \cite{Malan2023}, the voltage following buses are not \ac{EIP} (see \autoref{prop:Control:no_power_no_eip}). 
We therefore investigate an entire cluster in this section instead of just the isolated buses. Specifically, we propose an \ac{LMI} which is used to verify an \ac{OSEIP} property for a cluster (see \autoref{sec:Clustering:Passivity}) before analysing the stability for a DC microgrid comprising multiple interconnected \ac{OSEIP} clusters (see \autoref{sec:Clustering:Stability}).
\subsection{Cluster OS-EIP} \label{sec:Clustering:Passivity}
%

Consider a cluster $m \in \setM$ with the uncontrolled dynamics in \eqref{eq:Problem:cluster_dynamics_shifted} comprising the set of voltage setting buses $\setS[m]$, and the set of voltage following buses $\setF[m]$. Applying the controllers \eqref{eq:Control:Voltage_setting} and \eqref{eq:Control:Voltage_following} results in the shifted cluster dynamics \eqref{eq:Clustering:cluster_dynamics} at the top of the next page,
\begin{figure*}[!t]
	\begin{equation} \label{eq:Clustering:cluster_dynamics}
		\scalebox{0.8}{$
			\setUnderBraceMatrix{\mQ[m]\vxedot[m]}{\begin{bmatrix} 
					\mLS[m] \vIntSedot[m]\\ \mLS[m] \viSaedot[m] \\ \mCS[m] \vvSedot[m] \\ \vIntFedot[m] \\ \mLF[m] \viFaedot[m] \\ \mCF[m] \vvFedot[m] \\ \mLTx[m] \viTxedot[m]
			\end{bmatrix}} {=} \setUnderBraceMatrix{\mA[m]}{\begin{bmatrix}
					\matrix{0} & \matrix{0} & \momegaVMod[m] & \matrix{0} & \matrix{0} & \matrix{0} & \matrix{0} \\
					-\Ident[S_m] & -\mpsiIMod[m] & -\mpsiVMod[m] & \matrix{0} & \matrix{0} & \matrix{0} & \matrix{0} \\
					\matrix{0} & \Ident[S_m] & \mYS[m] {-} \mdeltaS[m] & \matrix{0} & \matrix{0} & \matrix{0} & \mES[m] \\
					\matrix{0} & \matrix{0} & \matrix{0} & -\mtau[m] & \matrix{0} & \mtau[m]\malpha[m] & \matrix{0} \\
					\matrix{0} & \matrix{0} & \matrix{0} & \Ident[F_m] & -\mbetaMod[m] & -\malphaMod[m] & \matrix{0} \\
					\matrix{0} & \matrix{0} & \matrix{0} & \matrix{0} & \Ident[F_m] & \mYF {-} \mdeltaF[m] & \mEF[m] \\
					\matrix{0} & \matrix{0} & -\mEST[m] & \matrix{0} & \matrix{0} & -\mEFT[m] & {-}\mRTx[m]
			\end{bmatrix}} \!\! \setUnderBraceMatrix{\vxe[m]}{\begin{bmatrix}
					\vIntSe[m] \\ \viSae[m] \\ \vvSe[m] \\ \vIntFe[m] \\ \viFae[m] \\ \vvFe[m] \\ \viTxe[m]
			\end{bmatrix}} {+} \setUnderBraceMatrix{\mB[m]}{\begin{bmatrix}
					\matrix{0} & \matrix{0} \\
					\matrix{0} & \matrix{0} \\
					\Ident[S_m] & \matrix{0} \\
					\matrix{0} & \matrix{0} \\
					\matrix{0} & \matrix{0} \\
					\matrix{0} & \Ident[F_m] \\
					\matrix{0} & \matrix{0} \\
			\end{bmatrix}} \!\! \setUnderBraceMatrix{\vuce[m]}{\begin{bmatrix}
					\viExSe[m] \\ \viExFe[m]
			\end{bmatrix}}, 
			\quad \vyce[m] = \mBT[m] \vxe[m] = \! \begin{bmatrix}
				\vvSe[m] \\ \vvFe[m]
			\end{bmatrix}\!$}
	\end{equation}
	\hrulefill
\end{figure*}
where the control parameters are stacked diagonally to obtain the matrices $\momegaVMod[m], \mpsiInt[m], \mpsiIMod[m], \mpsiVMod[m], \mdeltaS[m]$ for the buses in $\setS[m]$ and $\mtau[m], \malpha[m], \mbetaMod[m], \malphaMod[m], \mdeltaF[m]$ for the buses in $\setF[m]$. Using the shifted dynamics in \eqref{eq:Clustering:cluster_dynamics}, we now investigate an \ac{OSEIP} property for the cluster.
\begin{theorem} \label{thm:Clustering:Cluster_passivity}
	The shifted cluster dynamics \eqref{eq:Clustering:cluster_dynamics} are \ac{OSEIP} w.r.t.\ the input-output pair $(\vuce[m], \vyce[m])$, if there is a matrix $\mP[m] = \mPT[m] \posDef 0$ and a $\siOut > 0$ such that 
	\begin{subequations} \label{eq:Clustering:Cluster_passive_cond}
		\begin{align}
			&\mQ[m] \mP[m] \mA[m] + \mAT[m] \mP[m] \mQ[m] + \siOut\mB\mBT \negSemiDef 0, \\
			&\mQ[m] \mP[m] \mB[m] = \mB[m] .
		\end{align}
	\end{subequations}
\end{theorem}
\begin{proof}
	Consider the positive definite storage function
	\begin{equation} \label{eq:Clustering:Cluster_Storage}
		\sH[m](\vxe[m]) = \frac{1}{2} \, \vxeT[m] \mQ[m] \mP[m] \mQ[m] \vxe[m] ,
	\end{equation}
	for the positive definite diagonal matrix $\mQ[m]$ defined in \eqref{eq:Clustering:cluster_dynamics}.
	Evaluating the \ac{OSEIP} requirement $\cramped{\sHdot[m] \le \vyceT[m] \vuce[m] - \siOut\vyceT[m] \vyce[m]}$ for \eqref{eq:Clustering:cluster_dynamics} leads to
	\begin{multline} \label{eq:Clustering:Cluster_passivity_ineq}
		\!\!\!\!\frac{1}{2} \, \vxeT[m] \left( \mQ[m] \mP[m] \mA[m] + \mAT[m] \mP[m] \mQ[m] \right) \vxe[m] + \vxeT[m] \mQ[m] \mP[m] \mB[m] \vuce[m] \\ \le \vxeT \mB[m] \vuce[m] - \siOut\vxeT[m] \mB[m] \mBT[m] \vxe[m] ,
	\end{multline}
	from which \eqref{eq:Clustering:Cluster_passive_cond} follows directly.
\end{proof}
The \ac{LMI} in \autoref{thm:Clustering:Cluster_passivity} provides an efficient means for verifying the \ac{OSEIP} of a cluster numerically.
We note that the port of \eqref{eq:Clustering:cluster_dynamics} comprises the exogenous current input and bus voltage output of each bus in the cluster. This allows for interconnections at \emph{any} bus in the cluster.
Furthermore, verifying \autoref{thm:Clustering:Cluster_passivity} for a cluster demonstrates that the effects of non-monotone loads at buses without power available in steady state can be compensated by other voltage setting buses in the cluster.

Building on the \ac{OSEIP} of the DC cluster investigated in \autoref{thm:Clustering:Cluster_passivity}, we now consider the dynamics of the cluster when its input and output are zero.
\begin{proposition} \label{prop:Clustering:Cluster_Zero_State}
	The cluster dynamics \eqref{eq:Clustering:cluster_dynamics} are \ac{ZSD}.
\end{proposition}
\begin{proof}
	Setting $\vuce[m] \equiv \vec{0}$ and $\vyce[m] \equiv \vec{0}$ in \eqref{eq:Clustering:cluster_dynamics} yields the autonomous zero state dynamics
	%
	{\allowdisplaybreaks
	\begin{subequations} \label{eq:Clustering:zero_state_dynamics}
		\begin{align} 
			\label{eq:Clustering:zero_state_dynamics_IntS}
			\mLS[m] \vIntSedot[m] &= \vec{0}, \\
			\label{eq:Clustering:zero_state_dynamics_iSa}
			\mLS[m] \viSaedot[m] &= - \vIntSe[m] - \mpsiIMod[m] \viSae[m], \\
			\label{eq:Clustering:zero_state_dynamics_IntF}
			\vIntFedot[m] &= -\mtau[m] \vIntFe[m], \\
			\label{eq:Clustering:zero_state_dynamics_iFa}
			\mLF[m] \viFaedot[m] &=  \vIntFe[m] -\mbetaMod[m] \viFae[m], \\ 
			\label{eq:Clustering:zero_state_dynamics_iTx}
			\mLTx[m] \viTxedot[m] &= - \mRTx[m] \viTxe[m],
		\end{align}
	\end{subequations}}
	where $\vvSe[m] \equiv \vec{0}$ and $\vvFe[m] \equiv \vec{0}$.
	Since $\vIntSedot[m] = \vec{0}$ from \eqref{eq:Clustering:zero_state_dynamics_IntS} and since the error variables are shifted to the equilibrium, $\vIntSe[m] = \vec{0}$.
	It then follows from \eqref{eq:Clustering:zero_state_dynamics_iSa} that $\lim_{t\to\infty} \viSae[m] = \vec{0}$, since $\mpsiIMod[m] \posDef 0$ (see \eqref{eq:Control:modified_gain_psi_omega}).
	Similarly, from \eqref{eq:Clustering:zero_state_dynamics_IntF} and \eqref{eq:Clustering:zero_state_dynamics_iTx} it can be seen that $\vIntFe[m]$ and $\viTxe[m]$ converge to zero, respectively.
	Finally, \eqref{eq:Clustering:zero_state_dynamics_iFa} ensures that $\lim_{t\to\infty} \viFae[m] = \vec{0}$ since $\vIntFe[m]$ also converges to zero.
	%
\end{proof}
\begin{remark} \label{rem:Clustering:LMI_feasiblity}
	It can be observed that several factors influence the feasibility of the \ac{LMI} in \autoref{thm:Clustering:Cluster_passivity}:
	\begin{itemize}
		\item Larger damping factors in $\mdeltaS$ and $\mdeltaF$ improves the feasibility but slows the closed-loop dynamics.
		\item Larger load parameters in $\mYS$ and $\mYF$ (see \eqref{eq:Problem:load_eip} and \autoref{fig:Problem:Split_Load}) restrict the \ac{LMI} feasibility.
		\item Increasing the number of lines in $\setE[m]$ improves the \ac{LMI} feasibility at the cost of increasing the \ac{LMI} order.
		\item Larger line resistances in $\mRTx$ restrict the feasibility of the \ac{LMI}.
		\item Increasing the electrical distance between a voltage following bus in $\setF[m]$ and its closest voltage setting bus neighbour in $\setS[m]$ restricts the \ac{LMI} feasibility.
	\end{itemize}
	Whereas the first two points relate to passivity properties of the individual buses (see \eqref{eq:Control:V_set_Condition_1} in \autoref{thm:Control:V_setting_passive}), the last three points considers the topology of the cluster which is investigated more closely in the sequel.
\end{remark}
\begin{remark} \label{rem:Clustering:Robust_passivity}
	The \ac{OSEIP} result in \autoref{thm:Clustering:Cluster_passivity} can be made robust against parameter uncertainty by considering a range possible parameter values and extending the \ac{LMI} to a robust optimisation problem (see \cite[p.~32]{Amato2006}). Specifically, using a uniform upper bound for $\mYS$ and $\mYF$ ensures robust stability against any load with a smaller $\sY[n]$ value (see \eqref{eq:Problem:load_eip}). Furthermore, this robust stability extends to any time-varying load that does not exceed the lack of \ac{IFP} (see \autoref{rem:Problem:load_IFP}) associated with the verified upper bound for $\mYS$ and $\mYF$, as long as the load is also \ac{ZSD}.
\end{remark}
\begin{remark} \label{rem:Clustering:Cluster_size}
	Verifying \autoref{thm:Clustering:Cluster_passivity} requires solving an \ac{LMI} of order $3|\setN|+|\setE|$. A computationally efficient clustering of a microgrid will therefore limit the buses with voltage setting controllers to one per cluster, i.e., $|\setS[m]| = 1$. The remaining buses with voltage following controllers can then be assigned to the cluster with the voltage setting bus to which they are electrically closest.
\end{remark}
\subsection{Clustered Microgrid Stability} \label{sec:Clustering:Stability}
Considering the established \ac{OSEIP} and \ac{ZSD} properties for a cluster \eqref{eq:Clustering:cluster_dynamics} (see \autoref{thm:Clustering:Cluster_passivity} and \autoref{prop:Clustering:Cluster_Zero_State}), we now focus on the overall microgrid consisting of the clusters in $\setM$, the static monotone load functions \eqref{eq:Problem:shifted_passive_load} at the buses in $\setN$, and the lines in $\setT$ interconnecting them. For this interconnection, we investigate stability w.r.t.\ the equilibrium
\begin{equation} \label{eq:Clustering:Microgrid_Equilibrium}
	\left\{\begin{aligned}
		&\vxe[m] = \vec{0}, \;  \vuce[m] = \vec{0}, \quad m \in \setM, \\
		&\viSetTe = \vec{0}
	\end{aligned}\right.
\end{equation}
where $\vxe[m]$ and $\vuce[m]$ are defined in \eqref{eq:Clustering:cluster_dynamics} and $\viSetTe$ denotes the vector of the current $\sie[j]$ through the transmission lines $j \in \setT$ with the dynamics \eqref{eq:Problem:line_dynamics}.
\begin{figure}[!t]
	\centering
	\resizebox{0.6\columnwidth}{!}{\begin{tikzpicture}[thick,>=latex']	
	\def\sDistx{1.75}
	\def\sDisty{1.0}
	\def\sminW{40mm}
	\def\sminWCouple{10mm}
	\def\sminHsiso{10mm}
	\def\sminHmimo{14mm}
	
	\tikzstyle{systemElement}	= [fill=white, drop shadow={opacity=0.35}]
	\tikzstyle{box}				= [draw, rectangle, systemElement, minimum height=\sminHsiso]
	\tikzstyle{boxSISO}			= [box, minimum width=\sminW]
	\tikzstyle{boxCouple}		= [box, minimum width=\sminWCouple]
	\tikzstyle{sum}				= [draw, circle, systemElement, inner sep = 0pt, minimum size = 7pt]
	
	\def\dxS{.5} 
	\def\dxM{.6} 
	\def\dyM{.35} 
	
	\def\dxOffset{.7}
	\def\dyOffset{.35}
	\def\dxOffsetC{.75}


	\coordinate(cSysCluster);
	
	\path (cSysCluster) to ++(0,2*\sDisty) coordinate (cSysLoad);
	\path (cSysCluster) to ++(0,-2*\sDisty) coordinate (cSysLine);
	\path (cSysCluster) to ++(-2*\sDistx,0) coordinate (cSum);
	
	\path (cSysCluster) to ++(-2*\sDistx,-\sDisty) coordinate (cIncLeft);
	\path (cSysCluster) to ++(2*\sDistx,-\sDisty) coordinate (cIncRight);
	
	\node[boxSISO] at (cSysCluster) (SysCluster) {\large Clusters \eqref{eq:Clustering:cluster_dynamics} in $\setM$};
	\node[boxSISO] at (cSysLoad) (SysLoad) {\large Loads \eqref{eq:Problem:shifted_passive_load} in $\setN$};
	\node[boxSISO] at (cSysLine) (SysLine) {\large Lines \eqref{eq:Problem:line_dynamics} in $\setT$};
	
	\node[boxCouple] at (cIncLeft) (IncLeft) {\large$\mESetT$};
	\node[boxCouple] at (cIncRight) (IncRight) {\large$\mESetTT$};

	\node [sum] at (cSum) (Sum) {};

	
	\draw[->] (SysCluster.east) node[above right] {$\vyce$} -| (IncRight.north);
	\draw[->] (IncRight.south) |- (SysLine.east) node[above right] {$\vuSetTe$};
	\draw[->] (SysLine.west) node[above left] {$\vySetTe$} -| (IncLeft.south);
	\draw[->] (IncLeft.north) -- (Sum.south) node [below left] {$-$};
	\draw[->] (Sum.east) -- (SysCluster.west) node[above left] {$\vuce$};
	
	\draw[->] (SysCluster.east -| IncRight.north) |- (SysLoad.east) node[above right] {$\vuLe$};
	\draw[->] (SysLoad.west) node[above left] {$\vyLe$} -| (Sum.north) node [above left] {$-$};
	
	\node[{fill=black,circle,minimum size=2.5pt,inner sep = 0}] at (SysCluster.east -| IncRight.north) () {};

\end{tikzpicture}}
	\caption{Interconnection of the clusters with the external lines and the monotone load functions.}
	\label{fig:Clustering:interconnection_diagram}
\end{figure}
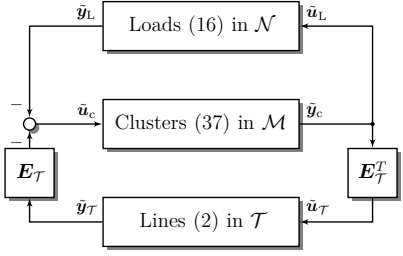
\begin{theorem} \label{thm:Clustering:Microgrid_stability}
	Consider a DC microgrid with $M$ clusters described by the dynamics \eqref{eq:Clustering:cluster_dynamics} along with the monotone load functions \eqref{eq:Problem:shifted_passive_load} at each bus $\setN$ which are interconnected by the lines \eqref{eq:Problem:line_dynamics} in $\setT$. Let \autoref{ass:Problem:Cluster_equilibrium} and the conditions of \autoref{thm:Clustering:Cluster_passivity} hold for each cluster $m \in \setM$. Then, the equilibrium \eqref{eq:Clustering:Microgrid_Equilibrium} of the DC microgrid is asymptotically stable.
\end{theorem}
\begin{proof}
	Consider the interconnection of the clusters, monotone load functions and lines as depicted in \autoref{fig:Clustering:interconnection_diagram}, where 
	\begin{equation}
		\mESetT \in \lbrace-1,0,1\rbrace^{\left|\bigcup_{m \in \setM}\setN[m]\right|\times|\setT|} 	
	\end{equation}
	is the incidence matrix describing the interconnection of the buses of the clusters in $\setM$ and the lines in $\setT$.
	Let $\vuceT = [\vuceT[1],\dots,\vuceT[M]]$ be the stacked input current errors, and $\vyceT = [\vyceT[1],\dots,\vyceT[M]]$ the stacked bus voltage errors for the clusters in $\setM$, with $\vuceT[m]$ and $\vyceT[m]$ as in \eqref{eq:Clustering:cluster_dynamics}. From \eqref{eq:Problem:bus_S_external_current} and \eqref{eq:Problem:bus_F_external_current}, the stacked input currents errors $\vuce$ is calculated from the currents of the lines in $\setT$ along with the loads at each bus, i.e.,
	%
	\begin{equation} \label{eq:Clustering:External_cluster_input}
		\vuce = -\mESetT \, \viSetTe - \viLPe(\vyce),
	\end{equation}
	where $\viLPe$ is obtained by stacking the load currents in \eqref{eq:Problem:shifted_passive_load}. We can then construct the static load system by defining
	\begin{equation} \label{eq:Clustering:Stacked_loads}
		\vuLe \coloneqq \vyce, \qquad \vyLe \coloneqq \viLPe(\vuLe).
	\end{equation}
	The stacked, linear dynamics of the lines in $\setT$
	\begin{equation} \label{eq:Clustering:cluter_connecting_lines}
		\mLSetT\, \viSetTedot = -\mRSetT\, \viSetTe + \vuSetTe, \quad \vuSetTe = \mESetTT\vyce, \quad \vySetTe = \viSetTe ,
	\end{equation}
	are \ac{OSEIP} w.r.t.\ the input-output port $(\vuSetTe, \vySetTe)$ \cite[Proposition~12]{Malan2023}.
	The closed-loop microgrid thus consists of the clusters as per \eqref{eq:Clustering:cluster_dynamics} with the exogenous current input in \eqref{eq:Clustering:External_cluster_input}, the monotone parts of the loads in \eqref{eq:Clustering:Stacked_loads}, and the lines in \eqref{eq:Clustering:cluter_connecting_lines}. Combining \eqref{eq:Clustering:External_cluster_input}--\eqref{eq:Clustering:cluter_connecting_lines} yields the interconnection matrix (see \autoref{fig:Clustering:interconnection_diagram})
	\begin{equation} \label{eq:Clustering:microgrid_interconnection}
		\begin{bmatrix}
			\vuce \\ \vuLe \\ \vuSetTe
		\end{bmatrix} = 
		\begin{bmatrix}
			\matrix{0} & -\Ident[|\setN|] & -\mESetT \\
			\Ident[|\setN|] & \matrix{0} & \matrix{0} \\
			\mESetTT & \matrix{0} & \matrix{0}
		\end{bmatrix}
		\begin{bmatrix}
			\vyce \\ \vyLe \\ \vySetTe
		\end{bmatrix} ,
	\end{equation}
	Asymptotic stability then follows since the clusters \eqref{eq:Clustering:cluster_dynamics} are \ac{OSEIP} and \ac{ZSD}, the monotone loads \eqref{eq:Clustering:Stacked_loads} are static and \ac{EIP} (see \autoref{rem:Problem:External_Load_EIP}), the lines \eqref{eq:Clustering:cluter_connecting_lines} are \ac{OSEIP} and \ac{ZSO}, and the interconnection matrix in \eqref{eq:Clustering:microgrid_interconnection} is skew symmetric\footnote{A matrix $\mA$ is skew symmetric if $\mA + \mAT = \matrix{0}$.} (see \cite[Proposition~4.4.15]{vdSchaft2017}).
\end{proof}
The results in \autoref{thm:Clustering:Microgrid_stability} thus verify the asymptotic stability of the DC microgrid containing non-monotone loads at arbitrary locations. Furthermore, we highlight that \autoref{thm:Clustering:Microgrid_stability} follows automatically if all clusters are \ac{OSEIP} (e.g.\ through \autoref{thm:Clustering:Cluster_passivity}).
\begin{remark} \label{rem:Clustering:Passive_cluster_framework}
	The use of passivity theory in \autoref{thm:Clustering:Microgrid_stability} allows for the direct inclusion of other \ac{OSEIP} buses (see \cite{Strehle2020dc,Nahata2020}) or buses with strictly monotone loads. In the context of \autoref{thm:Clustering:Microgrid_stability}, these buses may simply be regarded as clusters containing a single bus.
\end{remark}
\begin{remark} \label{rem:Clustering:Only_voltage_setting_buses}
	The \ac{OSEIP} of a cluster containing only a single controlled voltage setting bus, i.e., $|\setS[m]| = 1$, $|\setF[m]| = 0$, $|\setE[m]| = 0$, can be verified using \autoref{thm:Control:V_setting_passive} as opposed to the \ac{LMI} in \autoref{thm:Clustering:Cluster_passivity}. Correspondingly, the asymptotic stability of a DC microgrid comprising only \ac{OSEIP} controlled voltage setting buses, i.e., $|\setF[m]| = 0$, $\forall m \in \setM$, follows without requiring an \ac{LMI} to be solved.
\end{remark}

	\section{Robust Cluster Passivity} \label{sec:Analysis}
As demonstrated in the previous section, the microgrid stability hinges on verifying the \ac{LMI} in \autoref{thm:Clustering:Cluster_passivity} for each cluster. However, while the \ac{OSEIP} result in \autoref{thm:Clustering:Cluster_passivity} is robust against load changes (see \autoref{rem:Clustering:Robust_passivity}), it is not immediately robust against line parameters or cluster topologies changes, i.e., the \ac{LMI} must be re-evaluated when such changes occur. 

To relax this requirement and to achieve \autoref{obj:Problem:Cluster_Robustness}, we exploit the time-scales difference between the transmission line dynamics, which are characterised by $\sLTx[m,l]$, and the bus dynamics, which are characterised by $\sL[m,n]$ and $\sCeq[m, n]$. Thus, in \autoref{sec:Analysis:Singular_Perturb_Clusters}, we analyse a cluster using singular perturbation theory and show how the graph Laplacian matrix appears in the slow bus dynamics.
Thereafter, in \autoref{sec:Analysis:Laplacian}, we use properties of the Laplacian to make the passivity verification robust against various parameter and topology changes.
Note that in this section, for the sake of simplicity, we omit the cluster index $m$ where clear from context.

\subsection{Singularly Perturbed Clusters} \label{sec:Analysis:Singular_Perturb_Clusters}

As established in \cite{Calcev1999}, the passivity of a linear singular perturbed system is equivalent to the passivity of its slow and fast subsystems. The verification of \autoref{thm:Clustering:Cluster_passivity} can thus be simplified if the following assumption holds.
\begin{assumption} \label{ass:Analysis:Static_lines}
	The timescales of the closed-loop bus dynamics, \eqref{eq:Problem:bus_S_dynamics_passive_load} with \eqref{eq:Control:Voltage_setting} and \eqref{eq:Problem:bus_F_dynamics_passive_load} with \eqref{eq:Control:Voltage_following}, and the line dynamics \eqref{eq:Problem:line_dynamics} differ sufficiently to permit an analysis via singular perturbation theory.\footnote{As per \cite{Calcev1999}, we assume a suitable $0 < \epsilon \le \epsilon^*$ exists, i.e.\ that $\sLTx[m,l] \ll \sL[m,n], \sCeq[m, n]$.}
\end{assumption}
\begin{remark} \label{rem:Analysis:Static_lines_lit}
	\autoref{ass:Analysis:Static_lines} is similar to assuming static transmission lines (see \cite{Sadabadi2018,Tucci2018stab}) which is also based on a singular perturbation theory analysis as in \cite{Tucci2016}.
\end{remark}
With \autoref{ass:Analysis:Static_lines}, the controlled cluster dynamics in \eqref{eq:Clustering:cluster_dynamics} can be split into the fast system comprising the stacked transmission lines $l \in \setE$ from \eqref{eq:Problem:line_dynamics},
\begin{equation} \label{eq:Analysis:cluster_fast_dynamics}
	\mLTx \viTxedot = - \mRTx \viTxe - \mEST \vvSe - \mEFT \vvFe ,
\end{equation}
along with the slow system which is obtained by setting $\viTxedot = 0$ in \eqref{eq:Clustering:cluster_dynamics}, i.e.,
\begin{equation} \label{eq:Analysis:cluster_dynamics_reduced}
	\scalebox{0.82}{$
		\begin{aligned}
			\setUnderBraceMatrix{\mQr\vxredot}{\begin{bmatrix} 
				\mLS \vIntSedot\\ \mLS \viSaedot \\ \mCS \vvSedot \\ \vIntFedot \\ \mLF \viFaedot \\ \mCF \vvFedot
			\end{bmatrix}} {=}& \setUnderBraceMatrix{\mAr}{\begin{bmatrix}
				\matrix{0} & \matrix{0} & -\momegaVMod & \matrix{0} & \matrix{0} & \matrix{0} \\
				-\Ident[S] & -\mpsiIMod & -\mpsiVMod & \matrix{0} & \matrix{0} & \matrix{0} \\
				\matrix{0} & \Ident[S] & \mYS {-} \mdeltaS {-} \mLap[11] & \matrix{0} & \matrix{0} & {-} \mLap[12] \\
				\matrix{0} & \matrix{0} & \matrix{0} & -\mtau & \matrix{0} & \mtau\malpha \\
				\matrix{0} & \matrix{0} & \matrix{0} & \Ident[F] & -\mbetaMod & -\malphaMod \\
				\matrix{0} & \matrix{0} & {-} \mLapT[12] & \matrix{0} & \Ident[F] & \mYF {-} \mdeltaF {-} \mLap[22]\\
			\end{bmatrix}} \! \setUnderBraceMatrix{\vxre}{\begin{bmatrix}
				\vIntSe \\ \viSae \\ \vvSe \\ \vIntFe \\ \viFae \\ \vvFe
			\end{bmatrix}} \\&{+} \mBr \vuce , \\[5pt]
			\vyce \;{=}& \, \setUnderBraceMatrix{\mBrT}{\begin{bmatrix}
				\vec{0} & \vec{0} & \Ident[S] & \vec{0} & \vec{0} & \vec{0} \\
				\vec{0} & \vec{0} & \vec{0} & \vec{0} & \vec{0} &  \Ident[F]
			\end{bmatrix}} \, \vxre ,
		\end{aligned}$}
\end{equation}
where the graph Laplacian
\begin{equation} \label{eq:Analysis:Laplacian}
	\begin{bmatrix}
		\mLap[11] & \mLap[12] \\
		\mLapT[12] & \mLap[22]
	\end{bmatrix} \coloneqq \mLap = \mE \mRTxinv \mE^\Transpose, \qquad \mLap = \mLapT ,
\end{equation}
arises from the incidence matrix $\mE$ in \eqref{eq:Problem:interconnection} weighted by the line conductances $\mRTxinv$.
Note that the input $\vuce$ and output $\vyce$ are unchanged in the slow system \eqref{eq:Analysis:cluster_dynamics_reduced} and do not appear in the fast system \eqref{eq:Analysis:cluster_fast_dynamics}. Notice furthermore from \eqref{eq:Analysis:cluster_dynamics_reduced}, that the Laplacian can be separated from the rest of the slow dynamics $\mAb$ according to
\begin{equation} \label{eq:Analysis:Split_dynamics}
	\mAr = \mAb - \mBr \mLap \mBrT .
\end{equation}
The passivity of a cluster can now be analysed using the fast \eqref{eq:Analysis:cluster_fast_dynamics} and slow \eqref{eq:Analysis:cluster_dynamics_reduced} dynamics of the singular perturbed system.
\begin{proposition} \label{prop:Analysis:Reduced_cluster_passivity}
	A cluster \eqref{eq:Clustering:cluster_dynamics} for which \autoref{ass:Analysis:Static_lines} holds is \ac{OSEIP} w.r.t.\ the input-output pair $(\vuce, \vyce)$ and the storage function $\sHr(\vxre) = \vxreT \mQr \mPr \mQr \vxre$ if there exists a matrix $\mPr = \mPrT \posDef 0$ and a $\siOut > 0$ such that 
	\begin{equation} \label{eq:Analysis:Cluster_reduced_passive_cond}
		\mQr \mPr \mAr + \mArT \mPr \mQr + \siOut\mBr\mBrT \negSemiDef 0, \quad \mQr \mPr \mBr = \mBr .
	\end{equation}
\end{proposition}
\begin{proof}
	Consider first the linear fast dynamics \eqref{eq:Analysis:cluster_fast_dynamics}, where $\vvSe$ and $\vvFe$ are constant inputs. Because \eqref{eq:Analysis:cluster_fast_dynamics} is shifted to the equilibrium $(\vvSeq, \vvFeq, \viTxeq)$, it follows for constant input voltages that the error variables $\vvSe = \vec{0} = \vvFe$ and thus that $\vyce = \vec{0}$ in the fast system. For the fast dynamics with a storage function $\sHTx(\viTxe) = \viTxeT \mLTx \viTxe / 2$, and since $\mRTx \posDef 0$, it follows that
	\begin{equation}
		\sHTxdot(\viTxe) = - \viTxeT \mRTx \viTxe < 0 = \vyceT \vuce - \siOut\vyceT \vyce . 
	\end{equation}
	Thus, the fast system is \ac{OSEIP} w.r.t.\ $(\vuce, \vyce)$.
	
	Now consider the linear slow dynamics \eqref{eq:Analysis:cluster_dynamics_reduced} with the storage function $\sHr$, where the \ac{OSEIP} requirement $\cramped{\sHrdot \le \vyceT \vuce - \siOut\vyceT \vyce}$ reduces to \eqref{eq:Analysis:Cluster_reduced_passive_cond} in the same manner as in the proof of \autoref{thm:Clustering:Cluster_passivity}. Since the passivity of the singularly perturbed cluster is equivalent to the passivity of the slow and fast subsystems \cite[Theorem~1]{Calcev1999}, \eqref{eq:Analysis:Cluster_reduced_passive_cond} ensures the \ac{OSEIP} of the full cluster.
\end{proof}
\begin{remark} \label{rem:Analysis:reduced_LMI_size}
	The use of singular perturbation theory here allows the order of the \ac{LMI} verifying the cluster passivity to be reduced from $3|\setN|+|\setE|$ in \eqref{eq:Clustering:Cluster_passive_cond} to $3|\setN|$ in \eqref{eq:Analysis:Cluster_reduced_passive_cond}.
\end{remark}
\subsection{Passivity Preserving Topology Changes} \label{sec:Analysis:Laplacian}
The appearance of the Laplacian in the reduced cluster dynamics \eqref{eq:Analysis:cluster_dynamics_reduced} opens up the possibility of comparing different cluster configurations without needing to check \eqref{eq:Analysis:Cluster_reduced_passive_cond} for all of them. This can also be used, e.g., to determine if a cluster which undergoes a topology change retains its \ac{OSEIP} property.
\begin{proposition} \label{prop:Analysis:Cluster_comparison}
	Consider a cluster $m_1$ that satisfies \autoref{ass:Analysis:Static_lines} and the conditions of \autoref{prop:Analysis:Reduced_cluster_passivity}. Consider also a cluster $m_2$ with the same buses as cluster $m_1$, but with different lines. Let \autoref{ass:Analysis:Static_lines} also hold for the cluster $m_2$. Then, the cluster $m_2$ inherits the verified \ac{OSEIP} property of $m_1$ if
	\begin{equation} \label{eq:Analysis:Lap_comparison}
		\mLap[m_1] - \mLap[m_2] \negSemiDef 0 ,
	\end{equation}
	where $\mLap[m_1]$ and $\mLap[m_2]$ denote the Laplacian matrices of the respective clusters.
\end{proposition}
\begin{proof}
	Replace $\mAr$ in the \ac{LMI} \eqref{eq:Analysis:Cluster_reduced_passive_cond} with \eqref{eq:Analysis:Split_dynamics} to obtain
	\begin{multline} \label{eq:Analysis:Split_LMI}
		\mQr \mPr \mAb + \mAbT \mPr \mQr - \mQr \mPr \mBr \mLap \mBrT - \\ \mBr \mLapT \mBrT \mPr \mQr + \siOut\mBr\mBrT \negSemiDef 0 .
	\end{multline}
	The equality constraint in \eqref{eq:Analysis:Cluster_reduced_passive_cond} further reduces \eqref{eq:Analysis:Split_LMI} to 
	\begin{equation} \label{eq:Analysis:Reduced_split_LMI}
		\begin{aligned}
			\mQr \mPr \mAb + \mAbT \mPr \mQr - \mBr \mLap \mBrT - \mBr \mLapT \mBrT &\negSemiDef - \siOut\mBr\mBrT 
			\\ \mQr \mPr \mAb + \mAbT \mPr \mQr - 2\mBr \mLap \mBrT &\negSemiDef - \siOut\mBr\mBrT  ,
		\end{aligned}
	\end{equation}
	since $\mLap$ and thus $\mBr \mLap \mBrT$ are symmetric.
	Moreover, by assumption, the clusters $m_1$ and $m_2$ have the same buses, i.e., $\mAb[m_1] = \mAb[m_2]$, and where \eqref{eq:Analysis:Reduced_split_LMI} is verified for cluster $m_1$.
	Then, by comparing \eqref{eq:Analysis:Reduced_split_LMI} for the two clusters and setting
	\begin{multline} \label{eq:Analysis:Comparison_LMI}
			\mQr \mPr \mAb + \mAbT \mPr \mQr - 2\mBr \mLap[m_2] \mBrT \negSemiDef 
			\\ \mQr \mPr \mAb + \mAbT \mPr \mQr - 2\mBr \mLap[m_1] \mBrT \negSemiDef - \siOut\mBr\mBrT ,
	\end{multline}
	we can ensure that the cluster $m_2$ is \ac{OSEIP} through the \ac{OSEIP} of cluster $m_1$. The inequality \eqref{eq:Analysis:Comparison_LMI} in turn implies that $\mBr \mLap[m_1] \mBrT \negSemiDef \mBr \mLap[m_2] \mBrT$, from which \eqref{eq:Analysis:Lap_comparison} arises.
\end{proof}
\autoref{prop:Analysis:Cluster_comparison} thus allows the \ac{OSEIP} of a cluster with modified lines to be validated by computing the eigenvalues of an $|\setN|\times|\setN|$ matrix \eqref{eq:Analysis:Lap_comparison} instead of solving an \ac{LMI} of order $3|\setN|+|\setE|$. 
The following two corollaries show that verifying the \ac{OSEIP} of a modified cluster becomes even simpler if the values of the line resistances become smaller or if additional lines are added to the cluster.
\begin{corollary} \label{prop:Analysis:Modified_line_resistances}
	Consider a cluster $m_1$ for which the conditions of \autoref{prop:Analysis:Reduced_cluster_passivity} hold. Let a cluster $m_2$ have the same buses and lines as cluster $m_1$, but with different line resistances. Let \autoref{ass:Analysis:Static_lines} hold for both clusters. Then, the cluster $m_2$ inherits the verified \ac{OSEIP} properties of $m_1$ if
	\begin{equation} \label{eq:Analysis:Line_resistance_comparison}
		0 < \sRTx[m_2,l] \le \sRTx[m_1,l], \quad \forall l \in \setE[m_1] .
	\end{equation}
\end{corollary}
\begin{proof}
	Rewrite condition \eqref{eq:Analysis:Lap_comparison} using \eqref{eq:Analysis:Laplacian} along with the fact that the two clusters have the same incidence matrices $\mE[m_1] = \mE[m_2]$. This yields
	\begin{align} 
		\mE[m_1]\mRTxinv[m_1]\mET[m_1] - \mE[m_1]\mRTxinv[m_2]\mET[m_1] &\negSemiDef 0 \notag \\
		\label{eq:Analysis:Lap_comparison_same_topology}
		\implies \mRTxinv[m_1] - \mRTxinv[m_2] &\negSemiDef 0.
	\end{align}
	$\mRTx[m_1]$ and $\mRTx[m_2]$ being positive diagonal matrices, \eqref{eq:Analysis:Lap_comparison_same_topology} directly leads to \eqref{eq:Analysis:Line_resistance_comparison}.
\end{proof}
\begin{corollary} \label{prop:Analysis:Added_lines}
	Consider a cluster $m_1$ for which \autoref{ass:Analysis:Static_lines} and the conditions of \autoref{prop:Analysis:Reduced_cluster_passivity} hold. Then, a cluster $m_2$ obtained by duplicating $m_1$ and then adding additional lines between buses in $\setN[m_2]$ remains \ac{OSEIP}.
\end{corollary}
\begin{proof}
	For the modified cluster $m_2$, let $\mE[m_2] = [\mE[m_1], \mEn]$ and $\mRTx[m_2] = \Diag[\mRTx[m_1],\mRTxn]$ denote the incidence and resistance matrices, where $\mEn$ and $\mRTxn$ correspond to the added lines in $\setE[m_2] \setminus \setE[m_1]$. Then, the condition \eqref{eq:Analysis:Lap_comparison} in \autoref{prop:Analysis:Cluster_comparison} becomes 
	\begin{equation} \label{eq:Analysis:Lap_comparison_added_lines}
		\begin{aligned}
			\mE[m_1]\mRTxinv[m_1]\mET[m_1] - \begin{bmatrix}
				\mE[m_1] & \mEn
			\end{bmatrix} \!\! \begin{bmatrix}
				\mRTxinv[m_1] & \matrix{0} \\
				\matrix{0} & \mRTxninv
			\end{bmatrix} \!\! \begin{bmatrix}
				\mET[m_1] \\ \mEnT
			\end{bmatrix} &\negSemiDef 0 \\
			\mE[m_1]\mRTxinv[m_1]\mET[m_1] - \mE[m_1]\mRTxinv[m_1]\mET[m_1] - \mEn\mRTxninv\mEnT &\negSemiDef 0 \\
			-\mEn\mRTxninv\mEnT &\negSemiDef 0 ,
		\end{aligned}
	\end{equation}
	which always holds for positive line resistances.
\end{proof}
Through application of \autoref{prop:Analysis:Modified_line_resistances} and \autoref{prop:Analysis:Added_lines}, the cluster \ac{OSEIP} is retained when the values of the line resistances decrease and when adding new lines. The \ac{OSEIP} of a cluster can thus be verified robustly by considering the cluster topology with the highest line resistance values and the fewest possible lines. Thus, \autoref{prop:Analysis:Modified_line_resistances} also provides robustness against effects such as line loading, line ageing, line wear, and the ambient temperature which influence the line resistance. Note that these two corollaries can be combined in an arbitrary order.
\begin{remark} \label{rem:Analysis:Braess_Paradox}
	The conservation of \ac{OSEIP} in \autoref{prop:Analysis:Modified_line_resistances} and \autoref{prop:Analysis:Added_lines} when lowering line resistances or when adding lines does not contradict Braess' paradox (see e.g.\ \cite{Schaefer2022}). Whereas Braess' paradox pertains to the transmission limits of lines possibly being exceeded when new connections are introduced, we only consider the cluster \ac{OSEIP} and microgrid stability properties in this work.
\end{remark}
In the next proposition, we consider the special case where an \ac{OSEIP} cluster with a complete graph topology is used as a starting point. In this case, the comparison with clusters having different sets of edges reduces to comparing their algebraic connectivities $\seigLapTwo[m]$, i.e., the second smallest eigenvalues of the Laplacians.
\begin{proposition} \label{prop:Analysis:Alg_Conn_comparison}
	Consider two clusters $m_1$ and $m_2$ with the same buses and for which \autoref{ass:Analysis:Static_lines} hold. Let the topology of $m_1$ be a complete graph such that the conditions of \autoref{prop:Analysis:Reduced_cluster_passivity} hold. Then, the cluster $m_2$ with an arbitrary connected topology is \ac{OSEIP} if
	\begin{equation} \label{eq:Analysis:Laplacian_alg_conn_compare}
		\seigLapTwo[m_1] \le \seigLapTwo[m_2] 
	\end{equation}
	where $\seigLapTwo[m]$ is the algebraic connectivity of a cluster $m$.
\end{proposition}
\begin{proof}
	Consider the eigendecomposition of the $|\setN|\times|\setN|$ Laplacian of a weakly connected graph
	\begin{subequations} \label{eq:Analysis:Laplacian_eigendecomp}
		\begin{align}
			\mLap &= \mV \mEig \mVinv, &\quad \mV &= [\vOneCol[|\setN|], \mVI],\\
			\mEig &= \Diag[\seig[n]], &\quad 0 &= \seig[1] < \seig[2] \le \dots \le \seig[|\setN|] ,
		\end{align}
	\end{subequations}
	with $\mLap$ having a kernel consisting of the vector of ones $\vOneCol[|\setN|]$ and an image comprising $\mVI$.
	Consider now the requirement \eqref{eq:Analysis:Lap_comparison} in \autoref{prop:Analysis:Cluster_comparison} where the Laplacians are decomposed according to \eqref{eq:Analysis:Laplacian_eigendecomp}
	\begin{equation} \label{eq:Analysis:Decomposed_Laplacian_compare}
		\mV[m_1] \mEig[m_1] \mVinv[m_1] - \mV[m_2] \mEig[m_2] \mVinv[m_2] \negSemiDef 0 .
	\end{equation}
	Since a complete graph has an eigenvalue with multiplicity $|\setN| - 1$ \cite[p.~8]{Brouwer2012}, $\mVI[m_1]$ may comprise any arbitrary set of $\cramped{|\setN| - 1}$ linearly independent vectors with $\mVI[m_1] \perp \vOneCol[|\setN|]$.
	We therefore choose $\mVI[m_1] \coloneqq \mVI[m_2]$ to simplify \eqref{eq:Analysis:Decomposed_Laplacian_compare} to
	\begin{equation} \label{eq:Analysis:Laplacian_eigenvalue_compare}
		\seig[n](\mLap[m_1]) \le \seig[n](\mLap[m_2]), \quad n = 1, \dots, |\setN|.
	\end{equation}
	Furthermore, since cluster $m_1$ has $\seigLapTwo[m_1] = \seig[i](\mLap[m_1])$ for $i = 3,\dots,|\setN|$ and since $\seig[1] = 0$, \eqref{eq:Analysis:Laplacian_eigenvalue_compare} reduces to comparing the algebraic connectivity of the two clusters \eqref{eq:Analysis:Laplacian_alg_conn_compare}.
\end{proof}
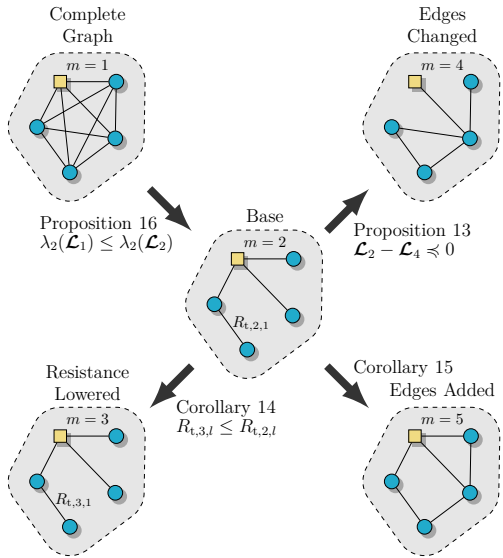
\begin{figure}[!t]
	\centering
	\resizebox{0.75\columnwidth}{!}{\begin{tikzpicture}
	\def\nodeXdist{1.5cm}
	\def\nodeYdist{1.15cm}
	\def\nodeXshift{\nodeXdist*0.75}
	\def\boundR{0.65cm}
	\def\clusterXdist{8.0cm}
	\def\clusterYdist{8.0cm}
	
	\def\Nodes{1,2,3,4,5}
	\def\NodesVSet{3/left}
	\def\NodesVFollow{1/below,2/left,4/above,5/right}
	
	\def\InternalEdgesComplete{1/2, 1/3, 1/4, 1/5, 2/3, 2/4, 2/5, 3/4, 3/5, 4/5}
	\def\InternalEdgesBase{1/2, 2/3, 3/4, 3/5}
	\def\InternalEdgesModified{1/2, 1/5, 2/5, 3/5, 4/5}
	\def\InternalEdgesResistance{1/2, 2/3, 3/4, 3/5}
	\def\InternalEdgesLineAdded{1/2, 2/3, 3/4, 3/5, 1/5, 4/5}
	
	\tikzstyle{nodeVSet}     		= [draw, rectangle, minimum width = 8pt, minimum height = 8pt, drop shadow={shadow scale=1.1,fill=black!35}, fill = buff]
	\tikzstyle{nodeVsetSmall}		= [draw, rectangle, minimum width = 4pt, minimum height = 4pt, drop shadow={shadow scale=1.1,fill=black!35}, fill = buff]
	\tikzstyle{nodeVFollow}  		= [draw, circle, drop shadow={shadow scale=1.1,fill=black!35}, radius=0.2pt, fill = ballblue]
	\tikzstyle{nodeVFollowSmall}	= [draw, circle, drop shadow={shadow scale=1.1,fill=black!35}, radius=0.1pt, fill = ballblue]
	
	\tikzstyle{lineInternal}		= [-]
	\tikzstyle{lineExternal}		= [-, looseness=0]
	\tikzstyle{backgroundCluster}	= [fill=black!10, draw, dashed]

	
	\coordinate(C1Mid);
	\path (C1Mid) ++(\clusterXdist, 0) coordinate (C4Mid) ++(0, -\clusterYdist)  coordinate (C5Mid) ++(-\clusterXdist, 0) coordinate (C3Mid) ++(\clusterXdist*0.5, \clusterYdist*0.5) coordinate (C2Mid);
	
	
	\path (C1Mid) ++(\nodeXshift*0.25, -\nodeYdist) coordinate (C1cNode1) ++(-\nodeXshift,\nodeYdist*0.5) coordinate (C1cNode2) ++(0,\nodeYdist) coordinate (C1cNode3) ++(\nodeXshift,\nodeYdist*0.5) coordinate (C1cNode4) ++(\nodeXshift*0.5,-\nodeYdist) coordinate (C1cNode5);
	\path (C2Mid) ++(\nodeXshift*0.25, -\nodeYdist) coordinate (C2cNode1) ++(-\nodeXshift,\nodeYdist*0.5) coordinate (C2cNode2) ++(0,\nodeYdist) coordinate (C2cNode3) ++(\nodeXshift,\nodeYdist*0.5) coordinate (C2cNode4) ++(\nodeXshift*0.5,-\nodeYdist) coordinate (C2cNode5);
	\path (C3Mid) ++(\nodeXshift*0.25, -\nodeYdist) coordinate (C3cNode1) ++(-\nodeXshift,\nodeYdist*0.5) coordinate (C3cNode2) ++(0,\nodeYdist) coordinate (C3cNode3) ++(\nodeXshift,\nodeYdist*0.5) coordinate (C3cNode4) ++(\nodeXshift*0.5,-\nodeYdist) coordinate (C3cNode5);
	\path (C4Mid) ++(\nodeXshift*0.25, -\nodeYdist) coordinate (C4cNode1) ++(-\nodeXshift,\nodeYdist*0.5) coordinate (C4cNode2) ++(0,\nodeYdist) coordinate (C4cNode3) ++(\nodeXshift,\nodeYdist*0.5) coordinate (C4cNode4) ++(\nodeXshift*0.5,-\nodeYdist) coordinate (C4cNode5);
	\path (C5Mid) ++(\nodeXshift*0.25, -\nodeYdist) coordinate (C5cNode1) ++(-\nodeXshift,\nodeYdist*0.5) coordinate (C5cNode2) ++(0,\nodeYdist) coordinate (C5cNode3) ++(\nodeXshift,\nodeYdist*0.5) coordinate (C5cNode4) ++(\nodeXshift*0.5,-\nodeYdist) coordinate (C5cNode5);

	\foreach \a/\b in \NodesVSet {
		\coordinate[rotate around={{-90+atan2(\nodeXshift,\nodeYdist*0.5)}:(C1Mid)}] (C1cNodeRot\a) at (C1cNode\a);
		\coordinate[rotate around={{-90+atan2(\nodeXshift,\nodeYdist*0.5)}:(C2Mid)}] (C2cNodeRot\a) at (C2cNode\a);
		\coordinate[rotate around={{-90+atan2(\nodeXshift,\nodeYdist*0.5)}:(C3Mid)}] (C3cNodeRot\a) at (C3cNode\a);
		\coordinate[rotate around={{-90+atan2(\nodeXshift,\nodeYdist*0.5)}:(C4Mid)}] (C4cNodeRot\a) at (C4cNode\a);
		\coordinate[rotate around={{-90+atan2(\nodeXshift,\nodeYdist*0.5)}:(C5Mid)}] (C5cNodeRot\a) at (C5cNode\a);
		\node[nodeVSet](C1node\a) at(C1cNodeRot\a) {};
		\node[nodeVSet](C2node\a) at(C2cNodeRot\a) {};
		\node[nodeVSet](C3node\a) at(C3cNodeRot\a) {};
		\node[nodeVSet](C4node\a) at(C4cNodeRot\a) {};
		\node[nodeVSet](C5node\a) at(C5cNodeRot\a) {};
	}
	\foreach \a/\b in \NodesVFollow {
		\coordinate[rotate around={{-90+atan2(\nodeXshift,\nodeYdist*0.5)}:(C1Mid)}] (C1cNodeRot\a) at (C1cNode\a);
		\coordinate[rotate around={{-90+atan2(\nodeXshift,\nodeYdist*0.5)}:(C2Mid)}] (C2cNodeRot\a) at (C2cNode\a);
		\coordinate[rotate around={{-90+atan2(\nodeXshift,\nodeYdist*0.5)}:(C3Mid)}] (C3cNodeRot\a) at (C3cNode\a);
		\coordinate[rotate around={{-90+atan2(\nodeXshift,\nodeYdist*0.5)}:(C4Mid)}] (C4cNodeRot\a) at (C4cNode\a);
		\coordinate[rotate around={{-90+atan2(\nodeXshift,\nodeYdist*0.5)}:(C5Mid)}] (C5cNodeRot\a) at (C5cNode\a);
		\node[nodeVFollow](C1node\a) at(C1cNodeRot\a) {};
		\node[nodeVFollow](C2node\a) at(C2cNodeRot\a) {};
		\node[nodeVFollow](C3node\a) at(C3cNodeRot\a) {};
		\node[nodeVFollow](C4node\a) at(C4cNodeRot\a) {};
		\node[nodeVFollow](C5node\a) at(C5cNodeRot\a) {};
	}
		
	\foreach \a/\b in \InternalEdgesComplete {
		\draw[lineInternal](C1node\a) to (C1node\b);	}
	\foreach \a/\b in \InternalEdgesBase {
		\draw[lineInternal](C2node\a) to (C2node\b);	}
	\foreach \a/\b in \InternalEdgesResistance {
		\draw[lineInternal](C3node\a) to (C3node\b);	}
	\foreach \a/\b in \InternalEdgesModified {
		\draw[lineInternal](C4node\a) to (C4node\b);	}
	\foreach \a/\b in \InternalEdgesLineAdded {
		\draw[lineInternal](C5node\a) to (C5node\b);	}
	
	\path (C1Mid) -- (C2Mid) node(cArrowStart1)[pos=0.35]{} node(cArrowEnd1)[pos=0.65]{};
	\path (C2Mid) -- (C3Mid) node(cArrowStart2)[pos=0.35]{} node(cArrowEnd2)[pos=0.65]{};
	\path (C2Mid) -- (C4Mid) node(cArrowStart3)[pos=0.35]{} node(cArrowEnd3)[pos=0.65]{};
	\path (C2Mid) -- (C5Mid) node(cArrowStart4)[pos=0.35]{} node(cArrowEnd4)[pos=0.65]{};
	
	\foreach \a in {1,2,3,4} {
		\draw[-{latex[round]}, black!80, fill=black!80, line width = 5](cArrowStart\a) -- (cArrowEnd\a);
	}


	\path (cArrowStart1) -- (cArrowEnd1) node[pos=0.7, above=5pt, align=left, anchor=north east]{\large \autoref{prop:Analysis:Alg_Conn_comparison}\\ \large $\seigLapTwo[1] \le \seigLapTwo[2]$};
	\path (cArrowStart2) -- (cArrowEnd2) node[pos=0.5, below=5pt, align=left, anchor=north west]{\large \autoref{prop:Analysis:Modified_line_resistances}\\ \large $\sRTx[3,l] \le \sRTx[2,l]$};
	\path (cArrowStart3) -- (cArrowEnd3) node[pos=0.5, below=5pt, align=left, anchor=north west]{\large \autoref{prop:Analysis:Cluster_comparison}\\ \large $\mLap[2] - \mLap[4] \negSemiDef 0$};
	\path (cArrowStart4) -- (cArrowEnd4) node[pos=0.5, above=12pt, align=left, anchor=west]{\large \autoref{prop:Analysis:Added_lines}};

	
	\path (C2node1) -- (C2node2) node[left=1pt,pos=0.6, align=left, anchor=west] {$\sRTx[2,1]$};
	\path (C3node1) -- (C3node2) node[left=1pt,pos=0.6, align=left, anchor=west] {$\sRTx[3,1]$};

	\begin{scope}[on background layer]
		
		\path[backgroundCluster] \convexpath{C1node1,C1node2,C1node3,C1node4,C1node5}{\boundR};
		\path[backgroundCluster] \convexpath{C2node1,C2node2,C2node3,C2node4,C2node5}{\boundR};
		\path[backgroundCluster] \convexpath{C3node1,C3node2,C3node3,C3node4,C3node5}{\boundR};
		\path[backgroundCluster] \convexpath{C4node1,C4node2,C4node3,C4node4,C4node5}{\boundR};
		\path[backgroundCluster] \convexpath{C5node1,C5node2,C5node3,C5node4,C5node5}{\boundR};
		
		\path (C1node3) -- (C1node4) node[above=0.7cm,midway,sloped,font=\large, align=center] {Complete\\ Graph};
		\path (C2node3) -- (C2node4) node[above=0.7cm,midway,sloped,font=\large] {Base};
		\path (C3node3) -- (C3node4) node[above=0.7cm,midway,sloped,font=\large, align=center] {Resistance\\ Lowered};
		\path (C4node3) -- (C4node4) node[above=0.7cm,midway,sloped,font=\large, align=center] {Edges\\ Changed};
		\path (C5node3) -- (C5node4) node[above=0.7cm,midway,sloped,font=\large, align=center] {Edges Added};
		
		\path (C1node3) -- (C1node4) node[above=0.15cm,midway,sloped] {$m = 1$};
		\path (C2node3) -- (C2node4) node[above=0.15cm,midway,sloped] {$m = 2$};
		\path (C3node3) -- (C3node4) node[above=0.15cm,midway,sloped] {$m = 3$};
		\path (C4node3) -- (C4node4) node[above=0.15cm,midway,sloped] {$m = 4$};
		\path (C5node3) -- (C5node4) node[above=0.15cm,midway,sloped] {$m = 5$};

	\end{scope}
\end{tikzpicture}}
	\caption{Overview and examples of the passivity preserving propositions and corollaries when changing the edges of a cluster.}
	\label{fig:Analysis:Cluster_comparisons}
\end{figure}
Through \autoref{prop:Analysis:Alg_Conn_comparison}, a cluster with a fictitious complete graph topology may be used to verify the \ac{OSEIP} of several clusters with the same number of buses but with different edges. Furthermore, the \ac{OSEIP} verification requires only a single eigenvalue of the Laplacians to be compared. The above propositions and corollaries dealing with the preservation of cluster \ac{OSEIP} when the graph edges are changed are summarised visually in \autoref{fig:Analysis:Cluster_comparisons}. 
\begin{remark} \label{rem:Analysis:Algebraic_Connectivity}
	The algebraic connectivity $\seigLapTwo$ may be computed using the Courant-Fischer theorem \cite{Fiedler1973, deAbreu2007}
	\begin{equation} \label{eq:Analysis:Comp_Algebraic_connectivity}
		\seigLapTwo = \min_{\vw \ne \vec{0}, \vw \perp \vOneCol} \frac{\vwT \mLap \vw}{\vwT \vw}.
	\end{equation}
	Additionally, $\seigLapTwo$ is known for certain graph types, like star, cycle, or path graphs \cite[p.~8]{Brouwer2012}\cite{Fiedler1973}, and bounds for $\seigLapTwo$ can be computed for more general graphs \cite[p.~52]{Brouwer2012}\cite{deAbreu2007}.
\end{remark}
\begin{remark} \label{rem:Analysis:Complete_graph_connectivity}
	When verifying \autoref{prop:Analysis:Reduced_cluster_passivity} for a complete graph, it is desirable to set the line resistances $\mRTx$ in \eqref{eq:Analysis:Laplacian} as large as possible. This leads to a smaller algebraic connectivity $\seig[2]$. The comparison \eqref{eq:Analysis:Laplacian_alg_conn_compare} in \autoref{prop:Analysis:Alg_Conn_comparison} will thus typically compare the connectivity of a complete graph with higher line resistances to a cluster containing fewer lines with lower resistances.
\end{remark}
\begin{remark} \label{rem:Analysis:Vertex_change_robustness}
	Through \autoref{thm:Clustering:Microgrid_stability}, the microgrid stability is maintained when \ac{OSEIP} clusters are added, removed or exchanged. Robustness against buses in a cluster connecting or disconnecting can thus be ensured by verifying the \ac{OSEIP} of the cluster before and after the bus (dis)connecting. Such robustness can be ensured practically by verifying \autoref{prop:Analysis:Reduced_cluster_passivity} for clusters with the differing numbers of buses and with connection topologies described by complete buses. Robustness for the actual cluster can then be verified by finding $\seigLapTwo$ and applying \autoref{prop:Analysis:Alg_Conn_comparison} to the state before and after the bus (dis)connecting.
\end{remark}

	\section{Simulation} \label{sec:Simulation}
%
%
%
In this section, we demonstrate the robust asymptotic stability of a DC microgrid with non-monotone loads. The simulation is conducted in MATLAB/Simulink using the Simscape electrical toolbox.
The code used to generate the simulations in this paper is available at \url{https://gitlab.kit.edu/albertus.malan/dc-clusters-sim}.
The simulation setup is described in \autoref{sec:Simulation:Setup}. Thereafter, we verify the \ac{EIP} of the clusters and the microgrid stability in \autoref{sec:Simulation:Clusters} and present the results in \autoref{sec:Simulation:Results}.
\subsection{Simulation Setup} \label{sec:Simulation:Setup}
%
We consider the network of 21 buses in \autoref{fig:Problem:Example_network} of which four buses are equipped with voltage setting controllers and the rest with voltage following controllers. The microgrid parameters are given in \autoref{tab:Simulation:Sim_params}, where uniform \ac{DGU} parameters are used for simplicity. All lines have the same per-length parameters and their lengths in \autoref{tab:App_Sim_Data:Line_Params} are randomly generated. Furthermore, the loads parameters are randomly generated according to the limits in \autoref{tab:Simulation:Sim_params}.\footnote{As an example, a $\SI{212.2}{\kilo\watt}$ P load with $\svCrit = \SI{266}{\volt}$ has $\sY[n] = 3$ (see \cite[Remark~9]{Malan2023}).} The exact load values used can be found in the simulation files under \url{https://gitlab.kit.edu/albertus.malan/dc-clusters-sim}.

The control parameters \autoref{tab:Simulation:Control_params} are chosen as follows. The control parameters $\spsiV$, $\spsiI$, and $\somegaV$ for the voltage setting control \eqref{eq:Control:Voltage_setting} are designed based on the guidelines in \cite{Ferguson2023}. Note that all setpoints $\svSP[j]$ are set to $\SI{380}{\volt}$.
The parameters for the voltage following control \eqref{eq:Control:Voltage_following} are tuned for sufficient damping of the voltage and current states via $\salpha$ and $\sbeta$, respectively. Furthermore, $\stau$ is chosen such that the voltage following controllers have a settling time of around $\SI{5}{\milli\second}$.
For both voltage controllers, $\sdelta$ is chosen to be at least an order of magnitude greater than $\sY[n]$. This ensures that the effects of non-monotone loads are sufficiently dominated by the controlled buses.

The test sequence for the microgrid is as follows.
%
\begin{itemize}
	\item $t=\SI{0}{\second}$: All buses start with $\sv[n] = \SI{0}{\volt}$ and are connected as in \autoref{fig:Problem:Example_network}. Random load parameters are chosen for each node.
	\item $t=\SI{0.2}{\second}$: Bus~14 depletes its local power supply and thus switches to a voltage following controller \eqref{eq:Control:Voltage_following}. Load parameters of $\SI{50}{\percent}$ of the buses selected at random are assigned new random values.
	\item $t=\SI{0.4}{\second}$: Power is again available at Bus~14, and switches back to the voltage setting controller \eqref{eq:Control:Voltage_setting}. All lines interconnecting the clusters are disconnected. Load parameters of $\SI{50}{\percent}$ of the buses selected at random are assigned new random values.
	\item $t=\SI{0.6}{\second}$: The clusters are reconnected as in \autoref{fig:Problem:Example_network}. All load parameters are set to $\sZinv[n] = \SI{-1}{\siemens}$, $\sI[n] = \SI{-10}{\ampere}$, $\sP[n] = \SI{-10}{\kilo\watt}$ $(\implies \sY[n] = 1.18)$, thus injecting significant power.
	\item $t=\SI{0.8}{\second}$: All load parameters are set to $\sZinv[n] = \SI{0}{\siemens}$, $\sI[n] = \SI{0}{\ampere}$, $\sP[n] = \SI{100}{\kilo\watt}$ $(\implies \sY[n] = 1.42)$, thus consuming significant power.
\end{itemize}
\begin{table}[!t]
	\scriptsize
	\centering
	\renewcommand{\arraystretch}{1.15}
	\caption{Simulation Parameters}
	\label{tab:Simulation:Sim_params}
	\begin{tabular}{@{\;}lr@{\;}lr@{\;}l@{\;}}
		\noalign{\hrule height 1.0pt}
		Voltages & $\svRef =$ & $\SI{380}{\volt}$ & $\svCrit =$ & $\SI{266}{\volt}$  \\
		\hline
		Elec. Lines \eqref{eq:Problem:line_dynamics} & $\sRTx[l] =$ & $\SI{0.01}{\ohm\per\kilo\meter}$ & $\sLTx[l] =$ & $\SI{2}{\micro\henry\per\kilo\meter}$ \\
		& $\sCTx[l] =$ & $\SI{22}{\nano\farad\per\kilo\meter}$ & length $\in$ & $[0.2; 5]\,\si{\kilo\meter}$ \\
		\hline
		ZIP Loads \eqref{eq:Problem:ZIP_model} & $|\sZinv| \le$ & $\SI{0.208}{\siemens}$ & $|\sI\:\!| \le$ & $\SI{79}{\ampere}$ \\
		& $|\sP\,| \le$ & $\SI{30}{\kilo\watt}$ & $\sY[n] \le$ & $3$ \\
		\hline
		DGU Filters & $\sR[n] =$ & $\SI{0.1}{\ohm}$ & $\sL[n] =$ & $\SI{1.8}{\milli\henry}$\\
		\eqref{eq:Problem:bus_S_dynamics_passive_load}, \eqref{eq:Problem:bus_F_dynamics_passive_load} & $\sC[n] =$ & $\SI{2.2}{\milli\farad}$ && \\
		\noalign{\hrule height 1.0pt}
	\end{tabular}
\end{table}
\begin{table}[!t]
	\scriptsize
	\centering
	\renewcommand{\arraystretch}{1.15}
	\caption{Rounded Line Lengths}
	\label{tab:App_Sim_Data:Line_Params}
	\begin{tabular}{c@{\quad}c@{\qquad}c@{\quad}c@{\qquad}c@{\quad}c@{\qquad}c@{\quad}c}
		\noalign{\hrule height 1.0pt}
		Line & $\si{\kilo\meter}$ & Line & $\si{\kilo\meter}$ & Line & $\si{\kilo\meter}$ & Line & $\si{\kilo\meter}$ \\
		\hline
		1--2 & $1.92$ & 1--7 & $2.49$ & 2--3 & $1.52$ & 3--4 & $1.76$ \\
		3--5 & $1.67$ & 4--15 & $1.57$ & 6--7 & $2.13$ & 7--8 & $2.69$ \\
		8--9 & $3.75$ & 8--11 & $4.44$ & 8--13 & $2.06$ & 9--10 & $2.62$ \\
		10--12 & $1.06$ & 11--12 & $1.65$ & 12--20 & $1.48$ & 13--14 & $3.35$ \\
		14--15 & $3.83$ & 14--16 & $1.20$ & 15--16 & $2.53$ & 16--18 & $0.99$ \\
		17--18 & $1.44$ & 17--20 & $3.31$ & 18--19 & $4.08$ & 18--21 & $1.17$ \\
		19--20 & $1.93$ & 20--21 & $1.31$  && && \\
		\noalign{\hrule height 1.0pt}
	\end{tabular}
\end{table}
\begin{table}[!t]
	\scriptsize
	\centering
	\renewcommand{\arraystretch}{1.15}
	\caption{Controller Parameters}
	\label{tab:Simulation:Control_params}
	\begin{tabular}{lr@{\;}l@{\quad}r@{\;}l@{\quad}r@{\;}l}
		\noalign{\hrule height 1.0pt}
		Voltage Setting \eqref{eq:Control:Voltage_setting} & $\somegaV =$ & $\SI{e6}{}$ & $\spsiV =$ & $25000$ & $\spsiI =$ & $250$ \\
		\hline
		Voltage Following \eqref{eq:Control:Voltage_following} & $\salpha =$ & $1$ & $\sbeta =$ & $20$ & $\tau =$ & $100$\\
		\hline
		Common Parameter & $\sdelta =$ & $100$ & & & & \\
		\noalign{\hrule height 1.0pt}
	\end{tabular}
\end{table}
\subsection{Cluster OS-EIP Verification} \label{sec:Simulation:Clusters}
To verify the \ac{OSEIP}\footnote{In each case, $\siOut = 0.001$ is selected for the cluster \ac{OSEIP}.} of the various clusters, we start by considering clusters with one voltage setting bus and $|\setN[m]|-1$ voltage following buses interconnected by a complete graph topology. The required algebraic connectivity $\seigLapTwo[m]$ that ensures \ac{OSEIP} of these fictitious clusters can then be found using \autoref{prop:Analysis:Reduced_cluster_passivity} and minimizing for $\seigLapTwo[m]$. The results for various cluster sizes are shown in \autoref{tab:Simulation:Cluster_complete}.
\begin{table}[!t]
	\scriptsize
	\centering
	\renewcommand{\arraystretch}{1.15}
	\caption{Minimum $\seig[2]$ for Complete Topology Clusters}
	\label{tab:Simulation:Cluster_complete}
	\begin{tabular}{lccccccc}
		\noalign{\hrule height 1.0pt}
		Buses & 3 & 4 & 5 & 6 & 7 & 8 & 9 \\
		\hline
		$\seigLapTwo$ & $9.3$ & $12.8$ & $16.6$ & $20.6$ & $24.9$ & $29.6$ & $34.6$ \\
		\noalign{\hrule height 1.0pt}
	\end{tabular}
\end{table}

The algebraic connectivities of the clusters in \autoref{fig:Problem:Example_network} are then computed using \eqref{eq:Analysis:Comp_Algebraic_connectivity}. To ensure for an additional robustness, the line resistances for the clusters are increased by $\SI{10}{\percent}$ for the \ac{OSEIP} verification steps. This yields $\seigLapTwo[3] = 30.5$ for the four bus cluster, $\seigLapTwo[1] = 28.4$ and $\seigLapTwo[4] = 67.1$ for the five bus clusters and $\seigLapTwo[2] = 12.0$ for the seven bus cluster. Comparing these values to \autoref{tab:Simulation:Cluster_complete} shows that Clusters~1, 3, and 4 are \ac{OSEIP} through \autoref{prop:Analysis:Alg_Conn_comparison}. However, since Cluster~2 does not meet condition \eqref{eq:Analysis:Laplacian_alg_conn_compare}, we instead verify its \ac{OSEIP} using \autoref{prop:Analysis:Reduced_cluster_passivity}. Thus, since all four clusters are \ac{OSEIP}, asymptotic stability follows from \autoref{thm:Clustering:Microgrid_stability}.

In the period $\SI{0.2}{\second} < t \le \SI{0.4}{\second}$, where Bus~14 can no longer supply power in steady state, the cluster compositions are modified as follows: Bus~15 is assigned to Cluster~1, Buses~13 and 14 to Cluster~2 and Bus~16 to Cluster~4. This yields the new connectivities $\seigLapTwo[1] = 20.8$ for Cluster~1 with six buses, $\seigLapTwo[2] = 11.9$ for Cluster~2 with nine buses, and $\seigLapTwo[4] = 32.5$ for Cluster~4 with six buses, and where Cluster~3 is dissolved entirely. Again, Clusters~1 and 4 are \ac{OSEIP} due to \autoref{prop:Analysis:Alg_Conn_comparison} (cf.\ \autoref{tab:Simulation:Cluster_complete}), whereas the \ac{OSEIP} of Cluster~2 is verified using \autoref{prop:Analysis:Reduced_cluster_passivity}. Asymptotic stability again follows from \autoref{thm:Clustering:Microgrid_stability}.

Finally, for $\SI{0.4}{\second} < t \le \SI{0.6}{\second}$, the clusters are disconnected to form four separated microgrids. Asymptotic stability of the four microgrids then follow directly from the \ac{OSEIP} and the \ac{ZSD} properties of the respective clusters.

Note that the \ac{OSEIP} and stability verifications are once-off and are performed offline.
\subsection{Results} \label{sec:Simulation:Results}
%
%
\begin{figure}[!t]
	\centering
	\resizebox{0.7\columnwidth}{!}{\includegraphics[scale=1]{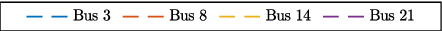}}\\[5pt]
	\resizebox{0.85\columnwidth}{!}{\includegraphics[scale=1]{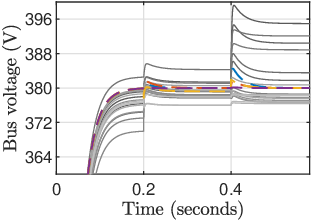} \hspace*{-5pt} \includegraphics[scale=1]{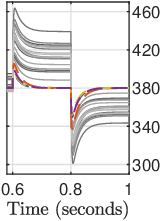}}
	\caption{Bus voltages of the DC microgrid. Voltage following buses are shown in grey.}
	\label{fig:Simulation:bus_voltage}
\end{figure}
\begin{figure}[!t]
	\centering
	\resizebox{0.85\columnwidth}{!}{\includegraphics[scale=1]{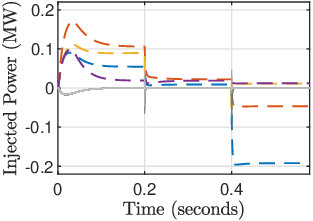} \hspace*{-5pt} \includegraphics[scale=1]{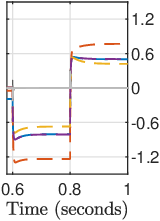}}
	\caption{Injected power at the buses. Voltage following buses are shown in grey.}
	\label{fig:Simulation:power_inject}
\end{figure}
\begin{figure}[!t]
	\centering
	\resizebox{0.85\columnwidth}{!}{\includegraphics[scale=1]{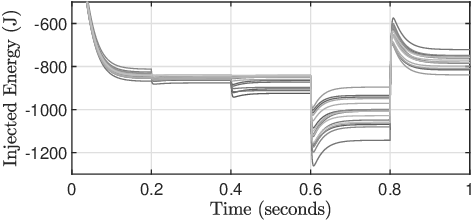}}
	\caption{Energy injected at the voltage following buses in set $\setF$.}
	\label{fig:Simulation:energy_inject}
\end{figure}
As predicted by the analysis, asymptotic stability is achieved for the entire duration of the simulation, as shown by the bus voltages in \autoref{fig:Simulation:bus_voltage}. For the realistic load values before $t = \SI{0.6}{\second}$, all bus voltages stay within $\SI{5}{\percent}$ of $\svRef$. Furthermore, the asymptotic stability is maintained for the periods with large negative loads, $\SI{0.6}{\second} < t \le \SI{0.8}{\second}$, and large positive loads, $\SI{0.8}{\second} < t \le \SI{1}{\second}$. Note that in these cases, the large differences between the bus voltages and $\svRef$ is a consequence of the network limitations. Nevertheless, the voltages at the voltage setting buses are successfully regulated to the setpoint $\svSP = \SI{380}{\volt}$. The power injected at the buses, shown in \autoref{fig:Simulation:power_inject}, similarly demonstrate that the voltage following buses inject zero power in steady state. Finally, although the peak power needed by the voltage following controllers are not insignificant, reaching a maximum of $\SI{0.3}{\mega\watt}$ at $t = \SI{0.8}{\second}$, the required energy, shown in \autoref{fig:Simulation:energy_inject}, remains small. Apart from the energy required for the start from $\SI{0}{\volt}$, the absolute energy deltas in the transient regions never exceed $\SI{570}{\joule}$ or $\SI{0.158}{\watt\hour}$.

	\section{Conclusion}\label{sec:Conclusion}
\acresetall

In this paper, we presented novel controllers for buses with and without steady-state power and which have uncertain loads exhibiting non-monotone incremental impedances. We proposed an \ac{LMI} to verify the \ac{OSEIP} of a cluster comprising both types of buses and showed that a microgrid comprising such clusters are asymptotically stable. Furthermore, by employing singular perturbation theory, we derived a reduced order \ac{LMI} for the cluster \ac{OSEIP} and showed that the \ac{OSEIP} is robust against certain parameter and topology changes.
In future work, we plan to integrate these results into networks comprising multiple energy domains.

	\appendices
	
	\section{Proof of \autoref{thm:Control:V_setting_passive}} \label{sec:App:Proof_V_setting_passive}
\begin{proof}
	The \ac{OSEIP} of \eqref{eq:Control:Shifted_voltage_setting_system} is verified if $\cramped{\sHsdot(\vxse) \le \syse\suse - \siOut \syse^2}$ and if $\sHs(\vxse) > 0$ for $\vxse \ne \vec{0}$ and $\sHs(\vec{0}) = 0$. Starting with the former, we calculate the time derivative $\sHsdot$ while noting that $\mQs\mPs  = \mPs\mQs$:
	\begin{align}
		\sHsdot(\vxse) &= \nicefrac{1}{2}(\vxsedot^\Transpose \mQs \mPs \vxse + \vxseT \mPs \mQs \vxsedot) \notag\\
		&= \nicefrac{1}{2}(\vxseT \mAsT \mPs \vxse + \vxseT \mPs \mAs \vxse) + \vxseT \vbs \sus .
	\end{align}
	Thus, $\sHsdot(\vxs) \le \syse\suse - \siOut \syse^2 = \vxseT \vbs \sus - \siOut \vxseT \vbs \vbsT \vxse$ if
	\begin{equation} \label{eq:Control:V_Setting_LMI}
		\setUnderBraceMatrix{-\mAsT\mPs - \mPs \mAs}{\begin{bmatrix}
				2 \spsiInt \sps[2] & \spsiInt\sps[3] + \spsiIMod\sps[2] & \spsiVMod \sps[2] - \somegaVMod\sps[1] \\
				\spsiInt\sps[3] + \spsiIMod\sps[2] & 2\spsiIMod\sps[3] & \spsiVMod\sps[3] - \somegaVMod\sps[2] - 1 \\
				\spsiVMod \sps[2] - \somegaVMod\sps[1] & \spsiVMod\sps[3] - \somegaVMod\sps[2] - 1 & 2(\sdelta - \sY - \siOut)
		\end{bmatrix}} \!\posSemiDef\! 0
	\end{equation}
	We verify the \ac{LMI} in \eqref{eq:Control:V_Setting_LMI} using Sylvester's criterion \cite[Theorem~7.2.5]{Horn2012}, by which $\det(\mAsT\mPs + \mPs \mAs) = 0$ and the last $N-1$ trailing principal minors must be positive. The determinant of \eqref{eq:Control:V_Setting_LMI} is zero if the three minors obtained by removing the last row are zero. Starting with the first minor,
	\begin{subequations}
		\begin{align}
			\det\begin{pmatrix}
				2\spsiInt\sps[2] & \spsiInt\sps[3] + \spsiIMod\sps[2] \\
				\spsiInt\sps[3] + \spsiIMod\sps[2] & 2\spsiIMod\sps[3] \\
			\end{pmatrix} &= 0 \\
			4\spsiInt \spsiIMod\sps[2]\sps[3] - (\spsiInt\sps[3] + \spsiIMod\sps[2])^2 &= 0 \\
			-(\spsiInt\sps[3] - \spsiIMod\sps[2])^2 &= 0,
		\end{align}
	\end{subequations}
	which holds if $\sps[3]$ is chosen as in \eqref{eq:Control:V_set_p_matrix}. For the second minor of the three minors, by substituting in $\sps[3]$ in \eqref{eq:Control:V_set_p_matrix}, we find
	\begin{subequations}
		\begin{align}
			\det\begin{pmatrix}
				2 \sps[2] & \spsiVMod \sps[2] - \somegaVMod\sps[1] \\
				\spsiInt\sps[3] {+} \spsiIMod\sps[2] & \spsiVMod\sps[3] {-} \somegaVMod\sps[2] {-} 1 \\
			\end{pmatrix} &= 0 \\
			\det\begin{pmatrix}
				2 \sps[2] & \spsiVMod \sps[2] - \somegaVMod\sps[1] \\
				2 \spsiIMod \sps[2] & \spsiIMod\spsiVMod\sps[2] - \somegaVMod\sps[2] - 1 \\
			\end{pmatrix} &= 0 \\
			2 \sps[2] \left(\spsiIMod\spsiVMod\sps[2] \,{-}\, \somegaVMod\sps[2] \,{-}\, 1\right) &= 2\spsiIMod\sps[2] \!\left(\spsiVMod \sps[2] \,{-}\, \somegaVMod\sps[1]\right) \\
			- \somegaVMod \sps[2] - 1 &= - \spsiIMod\somegaVMod\sps[1],
		\end{align}
	\end{subequations}
	where $\sps[2] > 0$ and from which the choice for $\sps[1]$ in \eqref{eq:Control:V_set_p_matrix} stems. By substituting $\sps[1]$ and $\sps[3]$ into the last of the three minors and simplifying, it can be verified that
	\begin{equation}
		\det\begin{pmatrix}
			\spsiInt\sps[3] + \spsiIMod\sps[2] & \spsiVMod \sps[2] - \somegaVMod\sps[1] \\
			2\spsiIMod\sps[3] & \spsiVMod\sps[3] - \somegaVMod\sps[2] - 1
		\end{pmatrix} = 0.
	\end{equation}
	Thus, the choices of $\sps[1]$ and $\sps[3]$ with $\sps[2] > 0$ ensure that $\det(\mAsT\mPs + \mPs \mAs) = 0$. The semi-definiteness of the matrix in \eqref{eq:Control:V_Setting_LMI} is then given if the last $N-1$ trailing principal minors are positive, i.e.,
	\begin{align}
		\label{eq:Control:V_set_trailing_minor_1}
		2(\sdelta - \sY - \siOut) &> 0, \\
		\label{eq:Control:V_set_trailing_minor_2}
		\det\begin{pmatrix}
			2\spsiIMod\sps[3] & \spsiVMod\sps[3] - \somegaVMod\sps[2] - 1 \\
			\spsiVMod\sps[3] - \somegaVMod\sps[2] - 1 & 2(\sdelta - \sY - \siOut)
		\end{pmatrix} &> 0.
	\end{align}
	From \eqref{eq:Control:V_set_trailing_minor_1}, the condition in \eqref{eq:Control:V_set_Condition_1} is derived. Calculating the determinant in \eqref{eq:Control:V_set_trailing_minor_2} and substituting $\sps[3]$ in \eqref{eq:Control:V_set_p_matrix} yields
	\begin{align}
		\setUnderBrace{\skappa[1]}{4 \spsiIMod^2 (\sdelta - \sY - \siOut)} \sps[2] - \left[\setUnderBraceMatrix{\skappa[2]}{\left(\somegaVMod - \spsiIMod\spsiVMod\right)} \sps[2] + 1\right]^2 &> 0 \notag \\
		\label{eq:Control:V_set_trailing_m_2_ineq}
		\skappa[2]^2 \sps[2]^2 + (2\skappa[2]-\skappa[1]) \sps[2] + 1 &< 0 ,
	\end{align}
	which is a quadratic inequality in $\sps[2]$. Since $\sps[2] > 0$, there will exist a $\sps[2]$ satisfying the inequality \eqref{eq:Control:V_set_trailing_m_2_ineq} if it has at least one positive root. This can in turn be investigated by ensuring that the turning point and the discriminant of the quadratic inequality are positive:
	{\allowdisplaybreaks
	\begin{align}
		\label{eq:Control:V_set_kappa_ineq_1}
		-\frac{(2\skappa[2]-\skappa[1])}{2\skappa[2]^2} > 0 &\implies \skappa[1] > 2\skappa[2], \\
		(2\skappa[2]-\skappa[1])^2 - 4\skappa[2]^2 > 0& \notag\\
		\label{eq:Control:V_set_kappa_ineq_2}
		-4\skappa[2]\skappa[1] + \skappa[1]^2 > 0	&\implies \skappa[1] > 4\skappa[2] ,
	\end{align}}
	Note that $\skappa[1] > 0$ because of \eqref{eq:Control:V_set_trailing_minor_1} and since $\sps[2] > 0$. Thus, \eqref{eq:Control:V_set_kappa_ineq_1} and \eqref{eq:Control:V_set_kappa_ineq_2} hold if $\skappa[2] \le 0$. If $\skappa[2] > 0$, \eqref{eq:Control:V_set_kappa_ineq_1} holds automatically if \eqref{eq:Control:V_set_kappa_ineq_2} is satisfied. By substituting $\skappa[1]$ and $\skappa[2]$ into \eqref{eq:Control:V_set_kappa_ineq_2}, the condition in \eqref{eq:Control:V_set_Condition_2} is obtained. Therefore, if conditions \eqref{eq:Control:V_set_Condition_1} and \eqref{eq:Control:V_set_Condition_2} hold, there is guaranteed to be a $\sps[2] > 0$ which, along with $\sps[1]$ and $\sps[3]$ in \eqref{eq:Control:V_set_p_matrix}, ensures that $\cramped{\sHsdot(\vxs) \le \syse\suse - \siOut\syse^2}$.
	
	Finally, we verify that $\sHs(\vxse) > 0$ for $\vxse \ne \vec{0}$. Since $\sL>0$ and $\sCeq>0$ and thus $\mQs \posDef 0$, it is sufficient to ensure that $\mPs \posDef 0$. Again using Sylvester's criterion \cite[Theorem~7.2.5]{Horn2012}, we verify $\mPs \posDef 0$ by ensuring that the leading principal minors of $\mPs$ are positive, i.e.,
	\begin{align}
		\label{eq:Control:V_set_positive_P_1}
		\sps[1] &> 0, \\
		\label{eq:Control:V_set_positive_P_2}
		\det\begin{pmatrix}
			\sps[1] & \sps[2] \\
			\sps[2] & \sps[3]
		\end{pmatrix} = \sps[1] \sps[3] - \sps[2]^2 &> 0 .
	\end{align}
	Since $\spsiIMod > 0$ and $\somegaVMod > 0$ hold by definition \eqref{eq:Control:modified_gain_psi_omega}, we see from \eqref{eq:Control:V_set_p_matrix} that $\sps[1] > 0$ and $\sps[3] > 0$ if $\sps[2] > 0$. Thus \eqref{eq:Control:V_set_positive_P_1} is met.
	Substituting $\sps[1]$ and $\sps[3]$ from \eqref{eq:Control:V_set_p_matrix} into \eqref{eq:Control:V_set_positive_P_2}, we obtain
	\begin{equation}
		\begin{aligned}
			\sps[2]^2 + \frac{1}{\somegaVMod} \sps[2] - \sps[2]^2 > 0 \implies
			\frac{1}{\somegaVMod} \sps[2] > 0,
		\end{aligned}
	\end{equation}
	which automatically holds since $\somegaV > 0$ and $\cramped{\sL > 0}$.
\end{proof}

	
	

 
	\bibliographystyle{IEEEtran}
	
	\bibliography{ms}
	
	

	%
\end{document}